\documentclass[12pt,english]{article}
\usepackage[T1]{fontenc}
\usepackage{lmodern}      
\usepackage{textcomp}     
\usepackage{microtype}    
\usepackage[utf8]{inputenc}
\usepackage{color}
\usepackage{mathtools}
\usepackage{amsmath}
\usepackage{amsthm}
\usepackage{amssymb}
\usepackage{bm} 
\usepackage[inline]{enumitem} 
\usepackage{geometry}
\usepackage{setspace}
\usepackage[authoryear]{natbib}
\usepackage{nameref}
\usepackage{comment}
\usepackage{booktabs}
\usepackage{pdflscape}
\usepackage{tikz}
\usepackage{pgfplots}
\usetikzlibrary{calc, arrows.meta, decorations.pathreplacing, positioning}
\pgfplotsset{compat=1.18}
\usepgfplotslibrary{groupplots}

\makeatletter
\theoremstyle{definition}
\newtheorem{rem}{\protect\remarkname}
\theoremstyle{definition}
\newtheorem{condition}{\protect\conditionname}
\theoremstyle{definition}
 \newtheorem{example}{\protect\examplename}
\theoremstyle{plain}
\newtheorem{cor}{Corollary}
\theoremstyle{plain}
\newtheorem{thm}{\protect\theoremname}
\theoremstyle{plain}
\newtheorem{lem}{\protect\lemmaname}

\usepackage{mleftright} 
\mleftright 

\usepackage[backref=page, hyperindex=true]{hyperref}
\hypersetup{%
    breaklinks=true,
    colorlinks=true,
    linkcolor=magenta,
    urlcolor=blue,
    citecolor=blue,
    pdfstartview={FitH},
    unicode=true
 }

\let \backreforig \backref
\renewcommand*{\backref}[1]{[\backreforig{#1}]}

\def\e{\mathrm{e}}          
\def\R{\mathbb{R}}          

\def\cl{\mathrm{cl}}        
\def\dom{\mathrm{dom}}      
\def\epi{\mathrm{epi}}      
\DeclareMathOperator{\interior}{int} 
\def\range{\mathrm{range}}  

\def\diag{\mathrm{diag}}      
\DeclareMathOperator{\rank}{rank} %

\def\ba{\boldsymbol{a}} 
\def\balpha{\boldsymbol{\alpha}} 
\def\bA{\boldsymbol{A}} 
\def\bb{\boldsymbol{b}} 
\def\bB{\boldsymbol{B}}
\def\bc{\boldsymbol{c}} 
\def\bD{\boldsymbol{D}} 
\def\bG{\boldsymbol{G}} 
 
\def\bI{\boldsymbol{I}} 
 
\def\bL{\boldsymbol{L}} 

\def\bN{\boldsymbol{N}} 
\def\bP{\boldsymbol{P}}
\def\bQ{\boldsymbol{Q}} 
\def\bR{\boldsymbol{R}} 
\def\bS{\boldsymbol{S}}
\def\bu{\boldsymbol{u}} 
\def\bU{\boldsymbol{U}} 
\def\bv{\boldsymbol{v}} 
\def\bw{\boldsymbol{w}} 
\def\bV{\boldsymbol{V}} 
\def\bx{\boldsymbol{x}} 
\def\by{\boldsymbol{y}} 
\def\bX{\boldsymbol{X}} 
\def\bz{\boldsymbol{z}} 
\def\blambda{\boldsymbol{\lambda}} 
\def\bSigma{\boldsymbol{\Sigma}} 
\def\bone{\mathbf{1}}   
\def\bzero{\mathbf{0}}  

\def\hatu{\widehat{u}}  
\def\obc{\overline{\bc}} 
\def\hx{\widehat{x}}    
\def\obx{\overline{\bx}} 
\def\tbx{\widetilde{\bx}} 
\def\hy{\widehat{y}}
\def\hbx{\widehat{\bx}} 
\def\hbz{\widehat{\bz}}

\def\cA{\mathcal{A}} 
\def\cC{\mathcal{C}} 
\def\cE{\mathcal{E}} 
\def\cG{\mathcal{G}} 
\def\cI{\mathcal{I}} 
\def\cL{\mathcal{L}} 
\def\cN{\mathcal{N}} 
\def\cS{\mathcal{S}} 
\def\cV{\mathcal{V}} 

\newcommand*{\argmin}{\operatornamewithlimits{argmin}\limits}
\newcommand*{\argmax}{\operatornamewithlimits{argmax}\limits}
\newcommand*{\Minimize}{\operatornamewithlimits{Minimize}\limits}
\newcommand*{\Maximize}{\operatornamewithlimits{Maximize}\limits}

\usepackage[affil-it]{authblk} 

\makeatother

\usepackage{babel}
\providecommand{\conditionname}{Condition}
\providecommand{\examplename}{Example}
\providecommand{\lemmaname}{Lemma}
\providecommand{\remarkname}{Remark}
\providecommand{\theoremname}{Theorem}

\begin{document}
\title{Convex Duality in Perturbed Utility Route Choice
}
\author[1]{Mogens Fosgerau}
\author[1,2]{Jesper R.-V.~S{\o}rensen}
\affil[1]{\small{University of Copenhagen}}
\affil[2]{\small{Aarhus Center for Econometrics (ACE)}}

\maketitle
\begin{abstract}
This paper develops a highly general convex duality framework for the perturbed utility route choice (PURC) model. We show that the traveler’s constrained, potentially non-smooth utility maximization problem admits a dual formulation: an unconstrained concave maximization problem with a differentiable objective. The unique optimal flow can be recovered link-by-link from any dual solution via the convex conjugates of link perturbation functions. These properties enable efficient gradient-based optimization for large-scale networks and fast computation for sensitivity analysis. Finally, the framework reveals a structural analogy between PURC and current flow in electrical circuits.
\end{abstract}

\section{Introduction}

Figure \ref{fig:GPS} illustrates a large dataset of GPS traces of bicycle trips in Copenhagen. Such data have become widely available, tracking millions of trips---by cars, trucks, bicycles, and using public transport---at high spatial and temporal resolution across large networks.
Modeling such data is essential for informing policy on the design and regulation of transport networks, including infrastructure investment and road pricing. This requires disaggregate models that predict individual route choices through large-scale networks.
    
\begin{figure}[htbp]
    \centering
    \begin{minipage}{0.57\textwidth}
        \centering
        \includegraphics[width=\linewidth]{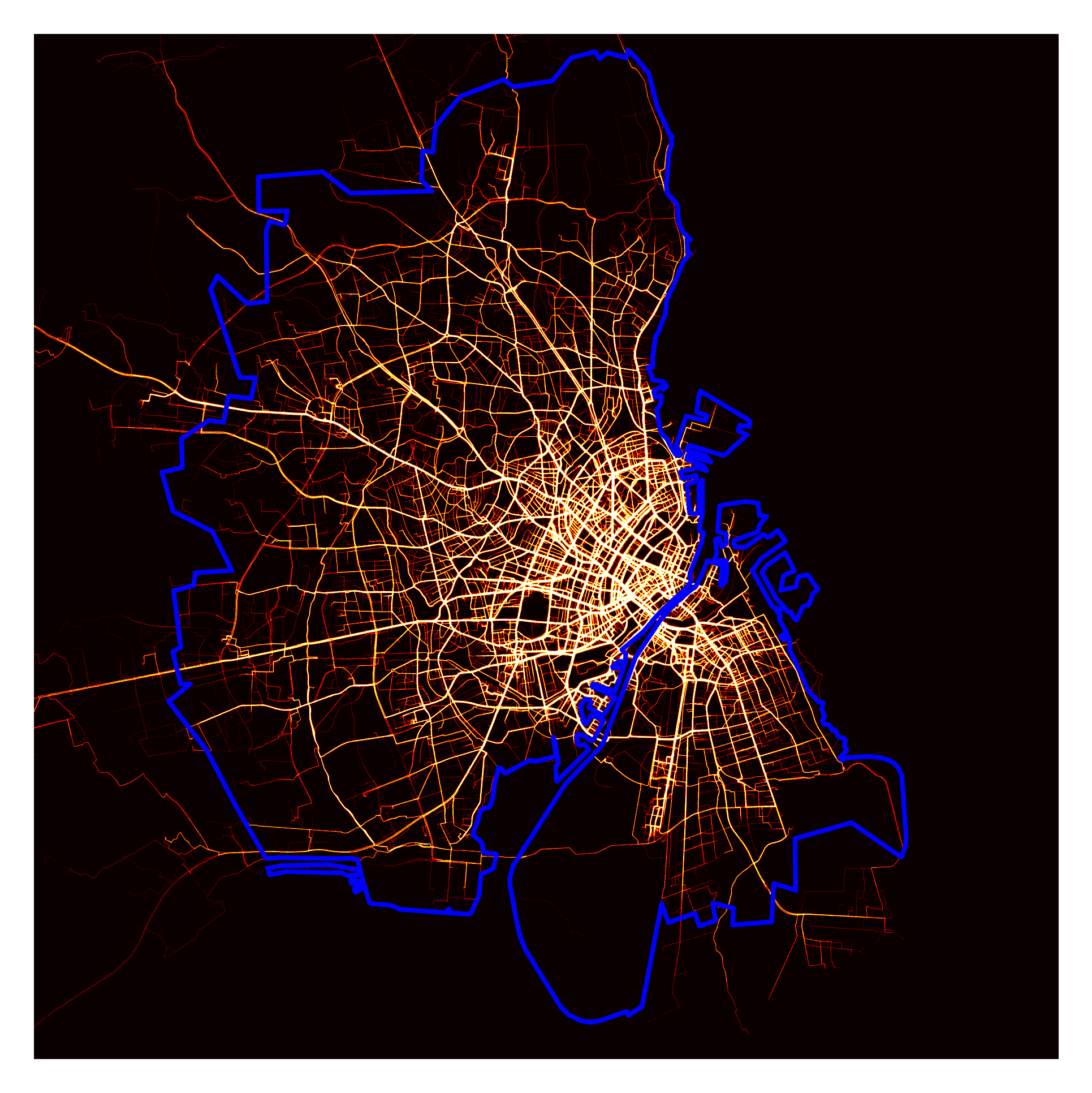}
    \end{minipage}
    \hfill
    \begin{minipage}{0.38\textwidth}
        \centering
        \includegraphics[width=\linewidth]{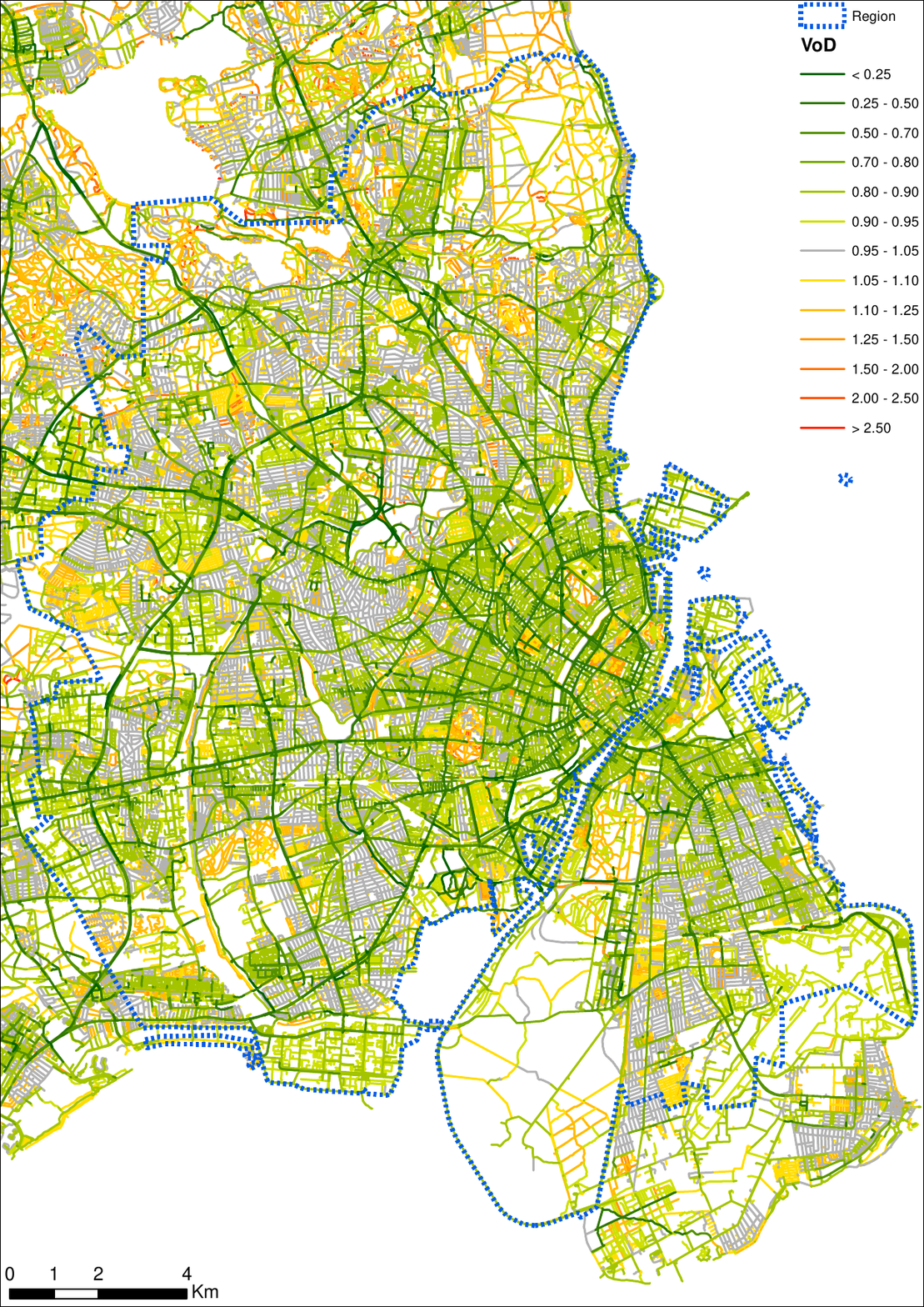}
    \end{minipage}
    \caption{Bicycle data and bicycle-relevant network for Copenhagen illustrated with figures from \citet{fosgerau2023bikeability}.}
    \label{fig:GPS}
\end{figure}

The sheer size and complexity of transportation networks makes this a challenging problem. For perspective, the number of distinct cycle-free paths between two points in an urban network is astronomically large---on the order of the square of the number of atoms in the universe \citep{roberts_estimating_2007}. Because of this combinatorial explosion, it has proven difficult to adapt the classical discrete choice model \citep{mcfadden_econometric_1981} to the route choice problem \citep{prato_route_2009}.

This paper develops a convex-analytic foundation for the perturbed utility route choice (PURC) model \citep{fosgerau2022perturbed}---a recent alternative to classical discrete choice approaches. Unlike traditional models that describe route choice as a discrete choice among a set of routes, PURC operates directly on the network structure, enabling tractable analysis and computation for large-scale systems. 

PURC is an instance of the general perturbed utility model proposed by \citet{McFadden2012}, which unifies several models widely used in economics \citep{Fudenberg2015, allen2019identification, allen2020hicksian, fosgerau_inverse_2024}. Building on this foundation, we derive a dual formulation and computational properties that make the model practical for large-scale applications \citep{fosgerau2023bikeability, yao_perturbed_2024}.

We show that the PURC model generalizes a model of electrical current flow \citep{rockafellar1984network}, placing it among physics-inspired transportation models such as the gravity model and the Lighthill-Whitham-Richards traffic flow model.

In the \emph{general perturbed utility model}, an agent maximizes utility of the form
\begin{equation*}
  U(\bx) = \langle \boldsymbol{v}, \bx \rangle - H(\bx), \quad \bx \in B,
\end{equation*}
where $\bx$ is a demand vector, $B \subseteq \R^n$ is a budget set, $\boldsymbol{v}$ is a vector of utility indexes, and $H$ is a convex function on $\R^n$ called the \emph{perturbation function}. Any additive random utility discrete choice model can be expressed in this form when the budget $B$ is the probability simplex \citep{Hofbauer2002}, in which case the demand $\bx$ is a probability vector.

The \emph{perturbed utility route choice} model describes route choice in a network from a given origin to a given destination. Unlike discrete choice models, demand is not a probability vector but a non-negative \emph{flow} assigned to each link. The flow on a link can be interpreted as the probability of observing a random trip on that link, and the budget constraint enforces flow conservation as traffic moves through the network. The perturbation function is separable across links,
\[
H(\bx) = \sum_{e \in \cE} h_e(x_e),
\]
where each $h_e$ is a univariate convex function. The perturbation function can be interpreted as an information cost, a Bregman divergence \citep{bregman_relaxation_1967} measuring the distance of the flow vector to some (here unspecified) base flow.

In this paper, we establish conditions for a general PURC model and derive properties essential for applications. Our main contributions are:
\begin{itemize}
    \item Existence and uniqueness of a solution to the traveler's PURC problem (the primal), ensuring a well-defined demand.
    \item A dual problem that is unconstrained, concave, and differentiable, even when the primal is non-smooth. These properties enable modern gradient-based optimization \citep[see, e.g.,][]{Nesterov2018} for large-scale networks.
    \item A closed-form dual objective with closed-form gradient and Hessian, enabling optimization and sensitivity analysis.
    \item A simple link-by-link recovery of demand from a dual solution via convex conjugates of link perturbation functions.
    \item A computationally simple expression of the Jacobian of the optimal flow useful for sensitivity analysis \citep{fosgerau_sensitivity_2025}.
    \item A clearer understanding of the model's duality structure through convex analysis.
    \item Understanding the PURC model as an analog of an electrical circuit. 
\end{itemize}

On the latter point, we find that a version of the perturbed utility structure is mathematically identical to the basic model of a resistive electric circuit \citep{rockafellar1984network}. Thus, flow corresponds to electric current, travelers to electrons, link perturbation functions $h_e$ to resistance, and flow conservation to Kirchhoff's current law. The Lagrange multipliers for flow conservation constraints act as potentials, reinforcing the analogy.  

The PURC framework applies to abstract networks. In addition to physical transport networks this includes, e.g., public transport networks \citep{de_cea_transit_1993}, dynamic networks with time dimensions \citep{ran_modeling_1996}, and hypernetworks incorporating mode, origin, destination, and timing choices \citep{sheffi_computation_1980}. It applies to activity-based models \citep{ben-akiva_activity_1998}, social and computer networks, and even educational or occupational pathways.

\subsection{Related Literature}
The literature on route choice is extensive; see \citet{prato_route_2009} for a review. Most studies treat route choice as a discrete choice among routes, which creates the challenge that the number of feasible routes is enormous. Research has therefore focused on constructing choice sets with good coverage to avoid bias and on developing models with realistic substitution patterns.

A newer line of research considers recursive models, where travelers choose paths link by link in a Markovian fashion. Early work addressed the assignment problem \citep{Dial1971, bell_alternatives_1995, shen_cyclic_1996, baillon_markovian_2008}. Subsequent studies introduced the recursive logit model and its generalizations based on the multivariate extreme value distribution \citep{fosgerau_link_2013, Mai2015, Mai2015a, mai_method_2016}. The duality properties of these models have been exploited by \citet{oyama_markovian_2022}. However, applying recursive logit to large networks remains challenging due to the need to invert large matrices. \citet{zimmermann_tutorial_2020} provides a comprehensive tutorial.

The PURC model \citep{fosgerau2022perturbed, fosgerau2023bikeability} departs from the discrete choice paradigm by modeling travelers as maximizing a perturbed utility function \citep{McFadden2012}. It requires no choice set and uses the entire network directly. Substitution patterns arise from the network structure via a perturbation function defined as a sum of convex link-flow functions. With aggregate route data, PURC can be estimated efficiently by linear regression. \citet{yao_perturbed_2024} propose a fast algorithm for computing equilibrium assignment in large networks, and \citet{fosgerau_sensitivity_2025} considers sensitivity analysis of the PURC equilibrium assignment.

Recently, the perturbed utility framework \citep{McFadden2012} has been proposed in the machine learning community by \citet{blondel_learning_2020}, who propose a regularized prediction function which is exactly the general perturbed utility demand function, and show that the induced Fenchel--Young loss provides a convex estimation criterion that remains well-defined when the model assigns zero probability to observed alternatives.

Several results in this paper build on the \citet{rockafellar1984network} book \emph{Network Flows and Monotropic Optimization}. For clarity and brevity, we opt for a self-contained presentation,  noting connections to \citet{rockafellar1984network} in Remark \ref{rem:Rocka84Relation} below. Related results can also be found in \citet{bertsekas1998network}.

\subsection{Roadmap}
The paper is organized as follows. Section \ref{sec:setup} introduces a general version of the PURC model and derives its main properties. Section \ref{sec:dual problem} presents the dual formulation, key theorems, and the analogy to electrical circuits. Section \ref{sec:applications} discusses efficient solution of the traveler's problem and sensitivity analysis. Section \ref{sec:example} provides a numerical example. Section \ref{sec:conclusion} concludes. Proofs are provided in the appendix along with key elements of convex analysis for easy reference.

\subsection{Notation}
We use the following notation: $\R = (-\infty,\infty)$ for the real numbers, $\R_{+} = [0,\infty)$ for the non-negative reals, and $\R_{++} = (0,\infty)$ for the strictly positive reals. Scalars are written as $x$, vectors as $\bx$, and matrices as $\bX$. For conformable vectors $\bx$ and $\by$, $\langle \bx,\by\rangle$ denotes their inner product, and $\left\|\bx\right\|$ the $\ell_2$/Euclidean norm. Vectors are treated as single-column matrices, and $\top$ indicates transposition. For a (twice) differentiable scalar function $f$, $\nabla f$ denotes its gradient, and $\nabla^2 f$ its Hessian.
For a differentiable vector function $\boldsymbol{f}$, $\nabla \boldsymbol{f}$ denotes its Jacobian. For a subdifferentiable convex function $f$, $\partial f$ denotes its subdifferential.

\section{Setup}\label{sec:setup}
This section introduces the model in three steps: Section \ref{subsec:Traffic-Network} defines the traffic network, Section \ref{subsec:Traveler-Preferences} specifies traveler preferences, and Section \ref{subsec:Traveler-Problem} formulates and analyzes the traveler’s problem.

\subsection{Traffic Network}\label{subsec:Traffic-Network}
We follow the terminology of \citet{rockafellar1984network}. A traffic \emph{network} $\cG$ consists of two sets, $\cV$ (nodes, of which there are at least two) and $\cE$ (links), and a function assigning to each link $e \in \cE$ a pair $(v,v')$ of distinct nodes. We write $e \sim (v,v')$ when $e$ connects nodes $v$ and $v'$, and denote $v(e)$ and $v'(e)$ as its initial and terminal nodes, giving the link an orientation that captures traffic direction. A link is \emph{incident} to its two nodes. Multiple (parallel) links that connect the same pair of nodes are allowed.

The sets of nodes and links are ordered, $\cV=\{ v_{1},\dots,v_{|\cV|}\}$ and $\cE=\{ e_{1},\dots,e_{|\cE|}\}$, which enables representations in terms of linear algebra. Thus, the network is represented by the \emph{incidence matrix} $\overline\bA=[a_{v,e}] \in \R^{\cV \times \cE}$, where\footnote{We follow the notation in PURC \citep{fosgerau2023bikeability}; \citet{rockafellar1984network} uses reversed signs.}

\[
a_{v,e}:=\begin{cases}
-1 & \text{if \ensuremath{v} is the initial node of the link \ensuremath{e},}\\
+1 & \text{if \ensuremath{v} is the terminal node of the link \ensuremath{e},}\\
\phantom{+}0 & \text{otherwise}.
\end{cases}
\]
A \emph{flow} in the network $\cG$ is a vector $\bx\in\R^{\cE}$, with the value $x_{e}$ being the flow on the link $e$. The product $\overline\bA\bx\in\R^\cV$ then contains the net inflow---the amount of flow arriving minus the amount departing---for each node.

A \emph{path} $P$ in the network $\cG$ is a finite sequence $v^{0},e^{1},v^{1},e^{2},\dotsc,e^{t},v^{t}$, where each $v^{k}$ is a node, each $e^{k}$ is a link, and either $e^{k}\sim(v^{k-1},v^{k})$ or $e^{k}\sim(v^{k},v^{k-1})$. The \emph{initial} node of the path $P$ is $v^{0}$ and the \emph{terminal} node is $v^{t}$, which we emphasize by writing $P:v^0\to v^t$. If the initial and terminal nodes coincide, then $P$ is called a \emph{circuit}. A link $e^{k}$ in $P$ is said to be traversed \emph{positively} or \emph{negatively} according to whether $e^{k}\sim(v^{k-1},v^{k})$ or $e^{k}\sim(v^{k},v^{k-1})$. If all links are traversed positively, then $P$ is called a \emph{positive} path (or circuit). The network is \emph{strongly connected} if every pair of distinct nodes is connected by a positive path. To ensure that travel between any origin and destination is compatible with the direction of traffic, we impose strong connectedness:
\begin{condition}
[\textbf{Connectedness}]\label{cond:Connectedness} The network $\cG$ is strongly connected.
\end{condition}

\subsection{Traveler Preferences}\label{subsec:Traveler-Preferences}
The traveler incurs a generalized (utility) cost for a flow vector $\bx \in \R^{\cE}$ given by
\[
\bx \mapsto \langle \bc, \bx \rangle + H(\bx),
\]
where $\langle \bc, \bx \rangle$ is a linear cost term and $H(\bx)$ is a perturbation cost term. The vector $\bc \in \R^{\cE}$ contains marginal costs $c_e$ for each link $e \in \cE$. The perturbation function $H: \R^{\cE} \to (-\infty,\infty]$ is separable across links:
\[
H(\bx) := \sum_{e \in \cE} h_e(x_e),
\]
where $\{h_e\}_{e \in \cE}$ are univariate \emph{link perturbation functions}. Recall that the \emph{effective domain} of a convex function $f:\R^n \to [-\infty,\infty]$ is the set
\[
\dom\,f := \left\{\bx \in \R^n \middle| f(\bx) < \infty\right\}.
\]
We impose the following condition on the link perturbation functions $\{h_e\}$:
\begin{condition}
[\textbf{Link Perturbation Functions}]\label{cond:Perturbed-Cost-Functions}
Each $h_{e},e\in\cE$, is a convex function from $\R$
to $\left(-\infty,\infty\right]$ with $\dom\,h_{e}=\R_{+}$,
which is continuous relative to $\R_{+}$, strictly convex
on $\R_{+}$, and satisfies $\lim_{\xi\to\infty}h_{e}\left(\xi\right)/\xi=\infty$.
\end{condition}
\noindent Interpreting Condition \ref{cond:Perturbed-Cost-Functions}, we note the following points:
\begin{itemize}
    \item The domain condition $\dom\,h_e = \R_{+}$ ensures that the perturbation cost is finite only when all link flows are non-negative, so $\dom\,H = \R_{+}^{\cE}$.
    \item Convexity and $\dom\,h_e = \R_{+}$ imply that $h_e$ is continuous on  $\R_{++}$, the interior of its effective domain. Requiring continuity relative to $\R_+$ is therefore equivalent to requiring lower semi-continuity at zero. Strict convexity guarantees uniqueness of the traveler’s solution (see Section \ref{subsec:Traveler-Problem}) and rules out linear graph segments; this requirement is equivalent to strict convexity on $\R_{++}$.
    \item The condition
    \(
    \lim_{\xi \to \infty} h_e(\xi)/\xi = \infty
    \)
    ensures that the marginal cost increases without bound as flow increases, thereby ensuring that the optimal flow is finite.
    \item Differentiability is not required, and thus the $h_e$ are allowed to have kinks. Convexity implies each $h_e$ is differentiable almost everywhere on $\interior(\dom\,h_e) = \R_{++}$.
\end{itemize}

\noindent We invoke Conditions \ref{cond:Connectedness} and \ref{cond:Perturbed-Cost-Functions} throughout the remainder of the paper.

\begin{example}
[\textbf{Link Perturbation Functions}]\label{exa:Cost-Functions} The following functions satisfy Condition \ref{cond:Perturbed-Cost-Functions}:
\begin{enumerate}[label=(\alph*)]
\item\label{exa:HalfQuadratic} The (half) \emph{quadratic} function
\[
h_{e}\left(\xi\right)=\begin{cases}
\infty & \text{if }\xi<0,\\
\frac{1}{2}\xi^{2} & \text{if }\xi\geqslant0.
\end{cases}
\]
\item\label{exa:NegativeEntropy} The (negative) \emph{entropy} function
\[
h_{e}\left(\xi\right)=\begin{cases}
\infty & \text{if }\xi<0,\\
0 & \text{if }\xi=0,\\
\xi\ln\xi & \text{if }\xi>0.
\end{cases}
\]
\item\label{exa:EntropyLike} The \emph{entropy-like} function
\[
h_{e}\left(\xi\right)=\begin{cases}
\infty & \text{if }\xi<0,\\
\left(1+\xi\right)\ln\left(1+\xi\right)-\xi & \text{if }\xi\geqslant0.
\end{cases}
\]
\item\label{exa:PiecewiseQuadratic} A (half) \emph{piecewise quadratic} function with a \emph{kink}  
\[
h_e(\xi) = 
\begin{cases} 
\infty & \text{if } \xi < 0, \\[6pt]
\frac{1}{2}\xi^2 & \text{if } 0 \leqslant \xi < 1, \\[6pt]
\frac{1}{2} + \frac{3}{2}(\xi - 1) + \frac{1}{2}(\xi - 1)^2 & \text{if } \xi \geqslant 1.
\end{cases}
\]
\end{enumerate}
These functions exhibit different dispersion properties. For example, the (negative) entropy function forces positive link flows because its derivative tends to $-\infty$ as $\xi \to 0_{+}$, whereas the quadratic, piecewise quadratic, and entropy-like functions allow zero flow on some links. Figure \ref{fig:perturbation-cost-functions} illustrates these functions.\hfill$\diamondsuit$
\end{example}

\begin{figure}[htbp]\label{fig:perturbation-cost-functions}
\centering
\includegraphics[width=\linewidth]{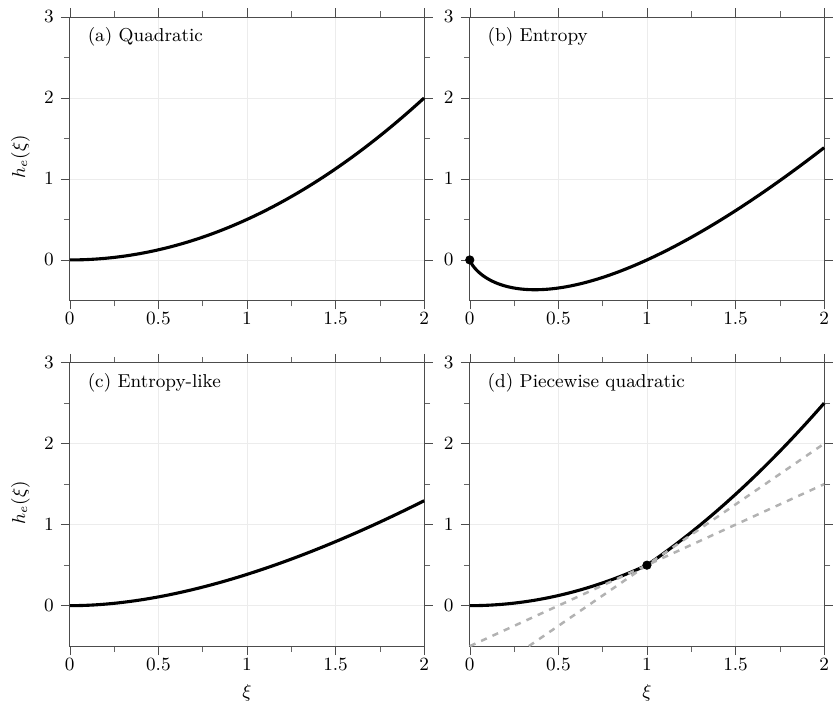}
\caption{Link perturbation functions from Example \ref{exa:Cost-Functions}. Panel (d) shows dashed gray tangents meeting at $\xi = 1$, highlighting the kink where two quadratic segments join.}
\end{figure}

\begin{rem}[\textbf{Relation to Existing PURC Literature}]
The previous PURC studies \citep{fosgerau2022perturbed, fosgerau2023bikeability, yao_perturbed_2024} assumed differentiable $h_e$ with $h_e(0)=0$, right derivative $\nabla^{+} h_e(0)=0$, and strictly positive link costs $\bc \gg \bzero$. These assumptions induce sparsity in optimal flows, which is both realistic and important for computational feasibility with large-scale networks. However, they are not required for the theoretical results presented here.\hfill$\diamondsuit$
\end{rem}

\subsection{The Traveler's Problem}\label{subsec:Traveler-Problem}

The traveler's \emph{demand} is a trip from an \emph{origin} $o\in\cV$ to a \emph{destination} $d\in\cV$, represented by a vector $\overline\bb\in\R^{\cV}$ defined by 
\[
\overline{b}_{v}:=\begin{cases}
-1 & \text{if \ensuremath{v} is the origin \ensuremath{o}},\\
+1 & \text{if \ensuremath{v} is the destination \ensuremath{d},}\\
\phantom{+}0 & \text{otherwise.}
\end{cases}
\]
Given the definition of the incidence matrix, the constraint $\overline\bA \bx = \overline\bb$ enforces flow conservation: net outflow from the origin and net inflow to the destination equal one, while all other nodes have zero net flow.
Flows satisfying $H(\bx) < \infty$ (i.e., $\bx \in \R_{+}^{\cE}$) and $\overline\bA \bx = \overline\bb$ are called \emph{feasible}. 

One can show that the row rank of $\overline\bA$ is $|\cV|-1$ \citep[][exercise 5.12]{bertsekas1998network}. This rank deficiency arises because the rows of $\overline\bA$ sum to zero, reflecting that total net inflow over all nodes must equal zero for any flow vector. We can therefore drop any single row and the corresponding entry from the incidence matrix $\overline\bA$ and demand vector $\overline\bb$, respectively. In what follows, we drop the row and entry corresponding to the destination node $d$, and use $\bA\in\R^{(\cV\setminus\{d\})\times\cE}$ and $\bb\in\R^{\cV\setminus\{d\}}$ to denote the resulting reduced incidence matrix and reduced demand vector, respectively. Let $\ba_{e}\in\R^{\cV\setminus\{d\}}$ denote the ``$e$th'' column of $\bA$.

The above modification is without loss of generality, in that a flow $\bx$ satisfies $\overline\bA\bx=\overline\bb$ if and only if $\bA\bx=\bb$. Economically, once flow conservation holds at all nodes except the destination, it must also hold at the destination because total outflow from the origin must equal total inflow to the destination.

The \emph{traveler's} (primal) \emph{problem} is then
\begin{equation}
\left.\begin{aligned}\Minimize\quad & \left\langle \bc,\bx\right\rangle +H\left(\bx\right)\text{ over flows }\bx\in\R^{\cE}\\
\text{satisfying}\quad & \bA\bx=\bb.
\end{aligned}
\right\} \tag{TP}\label{eq:Primal}
\end{equation}
Denote the \emph{optimal primal value}
\begin{equation}\label{eq:OptimalPrimalValue}
p:=\inf_{\substack{\bx\in\R^{\cE}:\\\bA\bx=\bb}}
\left\{ \left\langle \bc,\bx\right\rangle+H\left(\bx\right)\right\},
\end{equation}
and the \emph{optimal primal solution}
\[
\widehat{X}:=
\argmin_{\substack{\bx\in\R^{\cE}:\\
\bA\bx=\bb}}\left\{ \left\langle \bc,\bx\right\rangle+H\left(\bx\right)\right\},
\]
the set of minimizers in (\ref{eq:Primal}), if any. Our first result concerns the traveler's problem.
\begin{thm}
[\textbf{Primal Properties}]\label{thm:Primal-Properties} 
The traveler's problem (\ref{eq:Primal}) admits a unique optimal solution, i.e., $\widehat{X}=\left\{ \widehat{\bx}\right\}$ for some $\widehat{\bx}\in\R^{\cE}$, and the optimal primal value $p$ is finite and attained.
\end{thm}
\noindent
The proof of Theorem \ref{thm:Primal-Properties} first establishes existence of an optimal solution and then uses strict convexity to show uniqueness. Henceforth, we refer to $\widehat{\bx}$ as \emph{the optimal flow}.

Although the traveler’s problem is convex, $H$ may be non-differentiable due to kinks in $\{h_e\}$. Moreover, the flow conservation constraint $\bA \bx = \bb$ introduces an equation per (non-destination) node, which can number in the thousands for large networks, complicating the computation of the solution to the traveler's primal problem.

In the next section, we derive the traveler’s dual problem. We show that it is an unconstrained, differentiable, concave program with a closed-form objective, and that the optimal flow can be recovered from any dual solution.

\begin{rem}[\textbf{General Demand}]
The arguments for Theorem \ref{thm:Primal-Properties} (and similarly for Theorems \ref{thm:PstarIsPDual} and \ref{thm:Duality}) extend to any demand vector $\overline\bb \in \R^{\cV}$ satisfying $\sum_{v \in \cV} \overline{b}_v = 0$, allowing multiple origins and destinations. For simplicity, we focus on a unit demand between a single origin and destination.\hfill$\diamondsuit$
\end{rem}

\section{Traveler's Dual Problem and Duality}\label{sec:dual problem}
We analyze the conjugates $\{h_e^\ast\}$ and $H^\ast$ of the link perturbation functions $\{h_e\}$ and $H=\sum_e h_e$ in turn. For completeness, we first recall some key definitions from convex analysis; see Appendix \ref{sec:Convex-Analysis-Preliminaries} for further terminology and background.

The \emph{conjugate} $f^{\ast}:\R^n\to\left[-\infty,\infty\right]$ of a function $f:\R^n \to [-\infty,\infty]$ is given by
\begin{equation}\label{eq:DefinitionfConjugate}
f^{\ast}\left(\by\right):=\sup_{\bx\in\R^n}\left\{ \langle\bx,\by\rangle-f\left(\bx\right)\right\} ,\quad\by\in\R^n.
\end{equation}
The \emph{subdifferential of $f$ at the point $\bx$} is the (possibly empty) set
\[
\partial f\left(\bx\right):=\left\{ \by\in\R^n\middle|f\left(\bz\right)\geqslant f\left(\bx\right)+\langle\by,\bz-\bx\rangle)\text{ for all }\bz\in\R^n\right\} .
\]
The \emph{subdifferential of $f$} is the correspondence $\partial f:\R^n\rightrightarrows\R^n$, which has domain
\[
\dom\,\partial f:=\left\{ \bx\in\R^n\middle|\partial f\left(\bx\right)\neq\emptyset\right\} .
\]
We denote by $\dom\,\nabla f$ the set of points where $f$ is differentiable. For a (proper) convex function $f$, one always has $\dom\,\nabla f\subseteq\dom\,\partial f\subseteq\dom\,f$ \citep[Theorems 23.4 and 25.2]{rockafellar1970convex}.

The following lemma is key to our main results. 
\begin{lem}
[\textbf{Properties of Conjugate Link Perturbation Functions}]\label{lem:Conjugate-Perturbed-Cost-Contributions-Properties}
Each conjugate $h_e^\ast,e\in\cE$, is convex and differentiable on $\dom\, h^*_e=\R$. The derivative $\nabla h_{e}^{\ast}:\R\to\R$ is identified
by $\left\{ \nabla h_{e}^{\ast}(\eta)\right\} =\left(\partial h_{e}\right)^{-1}\left(\eta\right),\eta\in\R$,
and its range satisfies $\R_{++} \subseteq \range\,\nabla h_e^\ast = \dom\,\partial h_e \subseteq \R_{+}$. Moreover, $h_e^\ast$ admits the representation
\begin{equation}\label{eq:ConjugateExplicitFormula}
h_e^\ast(\eta) = (\partial h_e)^{-1}(\eta) \cdot \eta - h_e((\partial h_e)^{-1}(\eta)), \quad \eta \in \R,
\end{equation}
where $(\partial h_e)^{-1}(\eta)$ is interpreted as the unique scalar in the inverse subdifferential.
\end{lem}
\noindent
Lemma \ref{lem:Conjugate-Perturbed-Cost-Contributions-Properties} shows that although $\partial h_e$ is generally set-valued (e.g., for the quadratic case in Example \ref{exa:Cost-Functions}, $\partial h_e(0)=(-\infty,0]$), its \emph{inverse} is single-valued and equals $\nabla h_e^\ast$---a fact crucial for recovering optimal flows from dual variables.

The following example illustrates Lemma \ref{lem:Conjugate-Perturbed-Cost-Contributions-Properties} by computing explicit conjugates and their gradients for the link perturbation functions introduced earlier. These computations confirm that each conjugate is differentiable on $\R$ and that its gradient equals the inverse of the original subdifferential---a property central to the dual formulation.

\medskip\noindent\textbf{Example \ref{exa:Cost-Functions} (\nameref{exa:Cost-Functions}, Continued)}  
Direct computation gives:
\begin{enumerate}[label=(\alph*)]
\item Quadratic: $h_e^\ast(\eta)=\tfrac{1}{2}\max\{0,\eta\}^2$ for $h_e(\xi)=\tfrac{1}{2}\xi^2$ on $\R_{+}$;
\item Negative entropy: $h_e^\ast(\eta)=\mathrm{e}^{\eta-1}$ for $h_e(\xi)=\xi\ln\xi$ on $\R_{++}$ with $h_e(0):=0$;
\item Entropy-like: $h_e^\ast(\eta)=\max\{0,\mathrm{e}^\eta-\eta-1\}$ for $h_e(\xi)=(1+\xi)\ln(1+\xi)-\xi$ on $\R_{+}$;
\item Piecewise quadratic:
\[
h_e^\ast(\eta)=
\begin{cases}
0 & \text{if }\eta<0,\\
\tfrac{1}{2}\eta^2 & \text{if }0\leqslant\eta<1,\\
\eta-\tfrac{1}{2} & \text{if }1\leqslant\eta<\tfrac{3}{2},\\
\tfrac{1}{2}\eta^2-\tfrac{1}{2}\eta+\tfrac{5}{8} & \text{if }\eta\geqslant\tfrac{3}{2}.
\end{cases}
\]
\end{enumerate}
Each conjugate is differentiable. In particular, the latter conjugate has derivative
\[
\nabla h_e^\ast(\eta)=
\begin{cases}
0 & \text{if }\eta<0,\\
\eta & \text{if }0\leqslant\eta<1,\\
1 & \text{if }1\leqslant\eta<\tfrac{3}{2},\\
\eta-\tfrac{1}{2} & \text{if }\eta\geqslant\tfrac{3}{2},
\end{cases}
\]
even though $h_e$ has a kink.

Figure \ref{fig:conjugate-gradients} compares $\nabla h_e^\ast$ with $\partial h_e$. Dashed curves show $(\partial h_e)^{-1}$, obtained by reflecting the gradient graphs across the $45^\circ$ line. All gradients are continuous and non-decreasing, including the kinked case in Panel (d), where $\partial h_e(1)=[1,\tfrac{3}{2}]$ represents slopes of supporting hyperplanes at the kink (see Figure \ref{fig:perturbation-cost-functions}(d)).\hfill$\diamondsuit$

\begin{figure}[htbp]
  \centering
  \includegraphics[width=\linewidth]{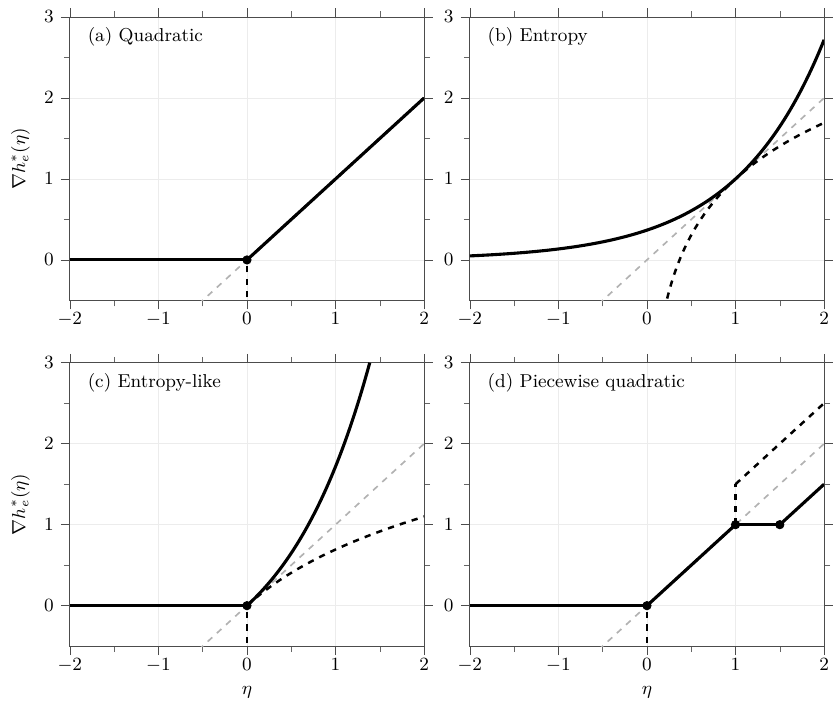}
  \caption{Gradients of conjugates $\nabla h_e^\ast$ (solid black). 
  The graph reflections in the 45 degree lines (dashed gray) depict the 
  subdifferentials $\partial h_e$ (dashed black).}
  \label{fig:conjugate-gradients}
\end{figure}

Using the conjugate link perturbations, we can derive the conjugate of the perturbation function $H$. The following lemma will also be convenient for our presentation.

\begin{lem}[\textbf{Properties of Conjugate Perturbation Function}]\label{lem:Conjugate-Perturbed-Cost-Properties}
The conjugate perturbation $H^*$ is the convex and differentiable function on $\dom\,H^*=\R^\cE$ given by 
\begin{equation}\label{eq:Conjugate-Perturbed-Cost-Explicit}
    H^*\left(\by\right) = \sum_{e\in \cE} h_{e}^\ast\left(y_e\right),\quad\by\in\R^\cE.
\end{equation}
Its gradient $\nabla H^*:\R^\cE\to\R^\cE$ is entrywise given by $[\nabla H^*(\by)]_e=\nabla h_{e}^{\ast
}(y_{e})$. Moreover, when it exists, its Hessian $\nabla^2 H^\ast(\by)$ at $\by$ is the diagonal matrix $\diag(\{\nabla^2 h_e^*(y_e)\}_{e\in\cE})\in\R^{\cE\times\cE}$.
\end{lem}

\medskip\noindent
We now formulate the traveler's dual problem. The \emph{Lagrangian} function $\cL$ \citep[see][p.~280]{rockafellar1970convex} associated with (\ref{eq:Primal}) is defined by
\begin{equation}\label{eq:LagrangianDefinition}
\cL(\bu,\bx):=\langle \bc,\bx\rangle+H(\bx)+\langle \bu,\bb-\bA\bx\rangle,\quad (\bu,\bx)\in\R^{\cV\setminus\{d\}}\times\R^{\cE},
\end{equation}
where $\bu\in\R^{\cV\setminus\{d\}}$ are Lagrange multipliers for the $|\cV|-1$ flow-conservation constraints $\bA\bx=\bb$. The corresponding \emph{dual function} $g:\R^{\cV\setminus\{d\}}\to\R$ is
\begin{equation}\label{eq:DualFunctionDefinition}
g(\bu):=\inf_{\bx\in\R^{\cE}}\cL(\bu,\bx),\quad \bu\in\R^{\cV\setminus\{d\}},
\end{equation}
and the traveler's (Lagrangian) dual problem is
\begin{equation}\label{eq:DualProblemProofs}
\Maximize\, g(\bu)\text{ over }\bu\in\R^{\cV\setminus\{d\}}.
\end{equation}
Following \citet{rockafellar1984network}, we interpret $\bu$ as a \emph{potential} vector on the network $\cG$: each $u_v$ assigns a potential to node $v\in\cV\setminus\{d\}$. For a link $e\sim(v,v')$ (not involving the destination), the potential difference $u_{v'}-u_v$ is the \emph{tension} across $e$. 
The tension on such a link is $\langle \ba_e,\bu\rangle$. We call any maximizer $\widehat{\bu}$ of \eqref{eq:DualProblemProofs} an \emph{optimal potential}.

Our next result gives a closed-form expression for the traveler's dual objective in terms of the conjugate link perturbation functions.

\begin{thm}[\textbf{Traveler's Dual Problem}]\label{thm:PstarIsPDual}
The dual of the traveler’s problem \eqref{eq:Primal} is
\begin{equation}\label{eq:Dual}
\Maximize\, g(\bu)=\langle \bb,\bu\rangle-H^*(\bA^\top \bu - \bc)\text{ over }\bu\in\R^{\cV\setminus\{d\}}.\tag{TP*}
\end{equation}
This problem is unconstrained with a differentiable concave objective function.
\end{thm}
\noindent
By Lemma \ref{lem:Conjugate-Perturbed-Cost-Properties}, $H^*$ is has full effective, domain $\dom\,H^\ast=\R^\cE$. Thus, \eqref{eq:Dual} imposes no explicit or implicit constraints on $\bu$. Since the conjugate $H^*$ is differentiable and convex, and $\bu$ enters only through affine terms, the objective in \eqref{eq:Dual} is differentiable and concave. Consequently, the traveler’s dual problem can be solved using any gradient-based algorithm for \emph{unconstrained concave maximization}.

Define the \emph{optimal dual value} as
\[
p^\ast:=\sup_{\bu\in\R^{\cV\setminus\{d\}}}g\left(\bu\right)
\]
and the \emph{optimal dual solution} as the set of maximizers (if any) in \eqref{eq:Dual},
\[
\widehat{U}:=\argmax_{\bu\in\R^{\cV\setminus\{d\}}}g\left(\bu\right).
\]
Our main result establishes the link between optimal primal and dual solutions.
\begin{thm}[\textbf{PURC Duality}]\label{thm:Duality}
The optimal dual value is attained and equals the optimal primal value: $p^\ast=p\in\R$. The set of optimal potentials $\widehat{U}$ is non-empty and closed. 
For any $\widehat{\bu}\in\widehat{U}$, the optimal flow $\widehat{\bx}$ is given by
\begin{equation}\label{eq:PrimalFromDual}
\widehat \bx = \nabla H^*(\bA^\top \widehat\bu -\bc).
\end{equation}
\end{thm}

\noindent
To interpret these findings, note that the dual objective $g$ in Theorem \ref{thm:PstarIsPDual} arises from a Lagrangian relaxation of the primal problem \eqref{eq:Primal}. This relationship implies \emph{weak duality}: $p^\ast\leqslant p$. A strict inequality $p^\ast<p$ (a \emph{duality gap}) can occur even in convex programs \citep[see][Example 5.3.2]{bertsekas2009convex}. The first part of Theorem \ref{thm:Duality} shows that PURC exhibits \emph{strong duality}: $p^\ast=p$.

The remaining statements involve the gradient of $g$:
\begin{equation}\label{eq:gradient}
\nabla g(\bu)=\bb - \bA \nabla H^*(\bA^\top \widehat\bu -\bc).
\end{equation}
The right-hand side represents \emph{excess demand}. Potentials act as \emph{prices} influencing traffic distribution, and any $\widehat{\bu}$ with $\nabla g(\widehat{\bu})=\bzero$ clears the ``market'' for traffic, meaning that flow conservation is achieved.

\begin{example}[\textbf{Optimal Flow via Duality in a Two-Node, Two-Link Network}]\label{exa:TwoNodeTwoLinkDuality}
Consider a network given by two nodes $\cV=\{v_1,v_2\}$ and two links $e_1\sim(v_1,v_2)$ and $e_2\sim(v_2,v_1)$. The traveler has origin $o=v_1$ and destination $d=v_2$. The marginals costs of travel on $e_1$ and $e_2$ are $c_{e_1}$ and $c_{e_2}$, respectively, and each link comes with a negative entropy perturbation, so that $h_e^*(\eta)=\nabla h_e^*(\eta)=\e^{\eta-1}$, cf.~Example \ref{exa:Cost-Functions}\ref{exa:NegativeEntropy}. The dual objective function is given by
\(
g(u)=u-\e^{u-c_{e_1}-1}-\e^{-u-c_{e_2}-1}.
\)
The first-order condition (FOC) for a maximum $\nabla g(u)=0$ reads
\(
\e^{u-c_{e_1}-1}-\e^{-u-c_{e_2}-1}=1,
\)
which using the change of variables $y:=\e^{u}$, can be written as the quadratic equation $\e^{-c_{e_1}-1} y^2-y-\e^{-c_{e_2}-1}=0$. With $\Delta(\bc):=1+4\e^{-c_{e_1}-c_{e_2}-2}>1$ denoting the discriminant, the (positive) solution is
\(
 \hy(\bc)=(1+\sqrt{\Delta(\bc)})(2\e^{-c_{e_1}-1}),
\)
which yields the (here unique) optimal potential
\[
\hatu\left(\bc\right):=1+c_{e_1}+\ln\left(\frac{1+\sqrt{\Delta\left(\bc\right)}}{2}\right).
\]
As $\nabla h_e^*(\eta)=\e^{\eta-1}$, from \eqref{eq:PrimalFromDual} we recover the optimal link flows as follows
\begin{equation}\label{eq:TwoNodeTwoLinkOptimalFlow}
\hbx\left(\bc\right)
=\begin{bmatrix}
    \nabla h_{e_1}^*\left(\hatu\left(\bc\right)-c_{e_1}\right) \\
    \nabla h_{e_2}^*\left(-\hatu\left(\bc\right)-c_{e_2}\right)
\end{bmatrix}
=\begin{bmatrix}
    \tfrac{1}{2}\left(\sqrt{\Delta(\bc)}+1\right)\\
    \tfrac{1}{2}\left(\sqrt{\Delta(\bc)}-1\right)
\end{bmatrix}
=\begin{bmatrix}
    \tfrac{1}{2}\left(\sqrt{1+4\e^{-c_{e_1}-c_{e_2}-2}}+1\right)\\
    \tfrac{1}{2}\left(\sqrt{1+4\e^{-c_{e_1}-c_{e_2}-2}}-1\right)   
\end{bmatrix}.
\end{equation}
Note that, no matter $\bc$, both links are active, and $\hx_{e_1}(\bc)$ exceeds one.\hfill$\diamondsuit$
\end{example}

\begin{rem}[\textbf{Relation to Network Flows}]\label{rem:Rocka84Relation}
An alternative proof of Theorems~\ref{thm:PstarIsPDual} and~\ref{thm:Duality} could follow Rockafellar’s network flow framework \citep{rockafellar1984network}:
\begin{itemize}
    \item Recast \eqref{eq:Primal} as an \emph{Optimal Distribution Problem} \citep[Section 8D]{rockafellar1984network}.
    \item Obtain \eqref{eq:Dual} via its dual, the \emph{Optimal Differential Problem} \citep[Section 8G]{rockafellar1984network}.
    \item Establish dual solvability using Rockafellar’s \emph{Existence Theorem for Optimal Potentials} \citep[p.~360]{rockafellar1984network}.
    \item Connect primal and dual solutions through his \emph{Network Equilibrium Conditions} \citep[p.~347]{rockafellar1984network}.
\end{itemize}
This approach would recover the optimal flow as in \eqref{eq:PrimalFromDual}. We adopt a more direct method based on Lagrangian relaxation and convex conjugacy to avoid additional notation. Finally, note that Rockafellar’s \emph{Network Equilibrium Theorem} \citep[p.~349]{rockafellar1984network} does not guarantee existence of optimal solutions for either problem; it must be shown separately.\hfill$\diamondsuit$
\end{rem}

\begin{rem}[\textbf{Identification Analogy in Discrete Choice}]
One can derive the dual problem in terms of the (complete) incidence matrix $\overline\bA$ and the (complete) demand vector $\overline\bb$, i.e., without dropping one row/entry. Carrying the additional and redundant constraint leads to a dual objective depending only on potential \emph{differences}.
Consequently, in this case, the set of optimal potentials $\widehat{\overline{U}}\subseteq\R^\cV$ will be not only non-empty and closed, but also \emph{unbounded}.

Such an indeterminacy in optimal potentials parallels random utility models (RUMs) in discrete choice, where utilities are identifiable only up to an additive constant. A common normalization is to fix a benchmark alternative (the \emph{outside option}); see \citet{sorensen2022mcfadden} for identification results in ARUMs. In our setting, potentials play a similar role to utilities, but since we focus on identifying the optimal \emph{flow}, no normalization is required for \eqref{eq:Dual}. However, imposing, e.g., $u_d=0$ (as we do), can be convenient for numerical implementation.\hfill$\diamondsuit$
\end{rem}

\subsection{Electrical Circuits Interpretation}
The PURC model can be viewed as a generalization of a standard electrical circuit model. Consider the network as an electrical circuit where each link has resistance given by the diagonal matrix $\bR\in\R^{\cE\times\cE}$. Let $\bc=\bzero$ and define the primal objective using quadratic link perturbations:
\[
H(\bx)=\tfrac{1}{2}\bx^\top\bR\bx=\tfrac{1}{2}\sum_{e\in\cE}R_{e,e}x_e^2.
\]
These perturbations do not satisfy Condition \ref{cond:Perturbed-Cost-Functions}, since their common domain is all of $\R$. Interpreting $\bx$ as current, the primal objective equals the network’s power dissipation. The flow conservation constraint $\bA\bx=\bb$ expresses Kirchhoff’s Law at all nodes except the destination node. Because the row corresponding to node $d$ was removed from the incidence matrix, this node plays the role of a grounded reference node. Dual variables $\bu$ represent electric potentials at the non-destination nodes. Since node $d$ was removed when forming $\bA$, these potentials are defined relative to the destination node, whose potential is normalized to zero. The vector $\bA^\top\bu$ collects potential (voltage) differences across links.

Solving this circuit problem yields the optimal flow/potential relationship
\[
\widehat{\bx}=\bR^{-1}\bA^\top\widehat{\bu},
\]
which is Ohm’s Law in disguise: current equals conductance times voltage difference. This analogy provides an intuitive interpretation of PURC: travel between an origin and destination resembles the movement of a unit probability mass through the network, distributing like electrons in an electric circuit.

A caveat is that electrical circuits allow negative currents, which simply indicate flow opposite to the chosen orientation of a link. However, in the traffic interpretation, negative flows are infeasible. This distinction motivates the domain restriction in Condition \ref{cond:Perturbed-Cost-Functions}. If instead one uses the half-quadratic perturbation from Example \ref{exa:Cost-Functions}\ref{exa:HalfQuadratic}, negative flows are ruled out and the general theoretical results of the paper apply.

\section{Applications of the Dual Objective}\label{sec:applications}

A typical application of the PURC model may involve networks with tens of thousands of links and traveler types. For example, a static forecast with $50{,}000$ links and $10{,}000$ traveler types requires solving the individual traveler’s problem $10{,}000$ times, each with $50{,}000$ variables. Estimation with micro-data is even more demanding, as individual flow predictions must be updated at every iteration of the estimator \citep{fosgerau_estimating_2024}.

Computation of the equilibrium assignment adds further complexity when links are congestible \citep{yao_perturbed_2024}. Travel times then depend on aggregate link flows, acting like prices in a general equilibrium model. Determining equilibrium involves solving for $50{,}000$ link ``prices'' and hundreds of millions of flows simultaneously---a formidable computational task.

The dual formulation \eqref{eq:Dual} is therefore a key contribution of this paper. We now show two applications that exploit its structure to accelerate computation.

\subsection{Efficient Solution of the Traveler's Problem}

We first derive the Hessian of the dual function $g$, recalling the expression in \eqref{eq:Dual}.

\begin{thm}[\textbf{Hessian of the Dual Function}]\label{thm:Hessian-Dual-Function} The dual objective $g$ is almost everywhere twice
differentiable with Hessian, where it exists, equal to 
\begin{equation}\label{eq:Hessian-Dual-Function}
\nabla ^{2}g\left( \bu\right) =-\bA \nabla^2 H^*( \bA^{\top }\bu-\bc)  \bA^{\top }.
\end{equation}
\end{thm}
\noindent
For computation, it may be useful to note that the Hessian of $g$ can be written as
\begin{equation*}
\nabla ^{2}g(\bu)=-\sum_{e\in \mathcal{E}}\nabla ^{2}h_{e}^{\ast
}(\langle \boldsymbol{a}_{e},\bu\rangle -c_{e})\boldsymbol{a}_{e}%
\boldsymbol{a}_{e}^{\top },
\end{equation*}%
where we have used Lemma \ref{lem:Conjugate-Perturbed-Cost-Properties}. These expressions are fast to compute, especially when many of the diagonal elements $\nabla ^{2}h_{e}^*(\langle \ba_{e},\bu\rangle -c_{e})$ are zero. The expressions enable efficient (damped) Newton steps for maximizing the dual objective. A numerically robust update computes a Newton direction $\Delta\bu$
by solving the (possibly sparse) linear system
\begin{equation*}
(-\nabla ^{2}g(\bu)) \Delta\bu=\nabla g(\bu),
\end{equation*}
thus avoiding explicit inversion of the Hessian.

\subsection{Sensitivity Analysis}
The aim of sensitivity analysis is to compute an approximation of the change in equilibrium flows and values following a change in network parameters without fully re-solving the equilibrium \citep{tobin_sensitivity_1988,yang_traffic_1997,patriksson_sensitivity_2003}. This is crucial for welfare analysis of network interventions \citep{Small1992}.

To this end, observe that Theorem \ref{thm:Primal-Properties} continues to hold for every $\bc\in\R^\cE$, since varying $\bc$ only modifies the linear term in the primal objective. We can therefore view the optimal flow $\hbx=\hbx(\bc)$ given in \eqref{eq:PrimalFromDual} as a function $\hbx:\R^\cE\to\R^\cE_+$ of the (marginal) link costs. We similarly view the optimal primal value $p=p(\bc)$ given in \eqref{eq:OptimalPrimalValue} as a function $p:\R^\cE\to\R$ of link costs. (That $p(\cdot)$ indeed takes values in $\R$ follows from Theorem \ref{thm:Primal-Properties}.) Finally, we define the \emph{cost--support function} $\cS:\R^\cE\to\cE$ as the identities of links with positive flow
\[
\cS\left(\bc\right):=\left\{e\in\cE\middle|\hx_e\left(\bc\right)>0\right\}
\]
as a function of $\bc$. The first result of this section gathers some properties of these functions. 

\begin{thm}[\textbf{Cost--Flow, Cost--Value and Cost--Support Properties}]\label{thm:CostFlowAndValueProperties} \begin{enumerate*}[label=(\arabic*)]
\item\label{enu:CostFlowContinuityAndMonotonicity} The cost--flow mapping $\bc\mapsto\widehat{\bx}(\bc)$ is a continuous function from $\R^\cE$ to $\R^\cE_+$ and anti-monotone in the sense that
\(
\langle\hbx\left(\bc\right)-\hbx\left(\bc'\right),\bc-\bc'\rangle\leqslant0\text{ for }\left(\bc,\bc'\right)\in\R^\cE\times\R^\cE.
\)
\item\label{enu:CostValueConcavity} The cost--value mapping $\bc\mapsto p(\bc)$ is a concave function from $\R^\cE$ to $\R$.
\item\label{enu:CostSupportConstancy} There is an open and dense subset $\cC_0$ of $\R^\cE$ on which the cost--support function $\bc\mapsto\cS(\bc)$ is locally constant: For each $\bc\in\cC_0$, there is an open neighborhood $\cN$ of $\bc$, such that $\bc'\in\cN$ implies $\cS(\bc')=\cS(\bc)$.
\end{enumerate*}
\end{thm} 

\noindent
That the optimal flow moves continuously with link costs is a basic sensitivity result. Moreover, consider a change in costs affecting only a single link $e\in\cE$, holding all other link costs fixed. In this case, anti-monotonicity (sometimes referred to as monotone decreasingness) reduces to
\[
\big(\hx_e(\bc)-\hx_e(\bc')\big)\cdot(c_e-c'_e)\leqslant0,
\]
since $\bc-\bc'$ has only one nonzero coordinate $(e)$. It follows that the \emph{corresponding} optimal flow on link $e$ is nonincreasing in its own cost, i.e.,
\[
c'_e>c_e \implies \hx_e\left(c'_e,\bc_{-e}\right)\leqslant \hx_e\left(c_e,\bc_{-e}\right),
\]
where $\bc_{-e}\in\R^{\cE\setminus\{e\}}$ denotes the vector of costs on all links other than $e$. This is a form of monotone comparative statics. 

Part \ref{enu:CostSupportConstancy} shows that there is an open set $\cC_0$ on which ``small'' cost perturbations have no impact on the identity of the active links (links with positive flow). As the set $\cC_0$ is dense in $\R^\cE$, one may interpret this finding as local constancy of the support holding generically in cost space.

\begin{rem}[\textbf{Cyclical Monotonicity}]
The proof of Theorem \ref{thm:CostFlowAndValueProperties} reveals that the cost--flow mapping $\bc\mapsto\hbx(\bc)$ is more than just monotone decreasing---it is \emph{cyclically anti-monotone}, meaning that $-\hbx\left(\cdot\right)$ is cyclically monotone \citep[see, e.g.,][p.~238]{rockafellar1970convex}. Monotonicity suffices for most purposes.\hfill$\diamondsuit$
\end{rem}
\noindent
To state our next result, recall that a proper convex function $f:\R^n\to(-\infty,\infty]$ is \emph{$\mu$-strongly convex} for some $\mu\in\R_{++}$ (with respect to the Euclidean norm), if for all $(\bx,\by)\in\dom\,f$ and all $\lambda\in[0,1]$,
\[
f\left(\left(1-\lambda\right)\bx+\lambda\by\right)\leqslant \left(1-\lambda\right)f\left(\bx\right)+\lambda f\left(\by\right)-\frac{\mu}{2}\lambda\left(1-\lambda\right)\left\Vert\bx-\by\right\Vert^2.
\]
(Equivalently, if $f-(\mu/2)\left\Vert\cdot\right\Vert^2$ is convex and $f$ is proper.) In what follows, we show that strongly convex link perturbation functions allow for a more precise sensitivity analysis.

\begin{thm}[\textbf{Cost--Flow Lipschitzness}]\label{thm:CostFlowLipschitzness}
In addition to Conditions \ref{cond:Connectedness} and \ref{cond:Perturbed-Cost-Functions}, for each $e\in\cE$, let $h_e$ be $\mu_e$--strongly convex. Then, $\bc\mapsto\hbx(\bc)$ satisfies
\begin{equation}\label{eq:CostFlowLipschitzWeighted}
\sqrt{\sum_{e\in\cE}\mu_e\left|\hx_e\left(\bc\right)-\hx_e\left(\bc'\right)\right|^2}\leqslant\sqrt{\sum_{e\in\cE}\frac{1}{\mu_e}\left|c_e-c'_e\right|^2}\text{ for }\left(\bc,\bc'\right)\in\R^\cE\times\R^\cE.     
\end{equation}
In particular, 
for $\mu_{\min}:=\min_{e\in\cE}\mu_e$, 
$\bc\mapsto\hbx(\bc)$ is $(1/\mu_{\min})$--Lipschitz continuous in that
\begin{equation}\label{eq:CostFlowLipschitzian}
\left\Vert\hbx\left(\bc\right)-\hbx\left(\bc'\right)\right\Vert\leqslant\frac{1}{\mu_{\min}}\left\Vert\bc-\bc'\right\Vert\text{ for }\left(\bc,\bc'\right)\in\R^\cE\times\R^\cE.    
\end{equation}
\end{thm}
\noindent
While Theorem \ref{thm:CostFlowAndValueProperties} only yielded qualitative conclusions such as monotonicity, strong convexity of the link perturbations leads to an explicit quantitative bound on changes in optimal flow. 

To illustrate, consider two cost vectors $\bc$ and $\bc'$ that differ only in coordinate $e$, with $c'_e>c_e$ and $c'_{e'}=c_{e'}$ for all $e'\neq e$. The weighted Lipschitz bound \eqref{eq:CostFlowLipschitzWeighted} then gives
\[
\sum_{e'\in\cE} \mu_{e'} |\hx_{e'}(\bc)-\hx_{e'}(\bc')|^2
\leqslant\frac{1}{\mu_e}|c_e-c'_e|^2 .
\]
Since each term in the sum is nonnegative, we obtain in particular
\[
\mu_e | \hx_e(\bc)-\hx_e(\bc') |^2
\leqslant
\sum_{e'\in\cE} \mu_{e'} |\hx_{e'}(\bc)-\hx_{e'}(\bc')|^2
\leqslant\frac{1}{\mu_e}|c_e-c'_e|^2,
\]
which implies the coordinate-wise bound for link $e$
\[
|\hx_e(\bc)-\hx_e(\bc')|
\leqslant
\frac{1}{\mu_e}|c_e-c'_e|.
\]
Combining this with the monotonicity property from Theorem \ref{thm:CostFlowAndValueProperties} yields
\[
0\leqslant
\frac{\hx_e(\bc)-\hx_e(\bc')}{c'_e-c_e}
\leqslant\frac{1}{\mu_e}.
\]
Bounds of this type allow for computationally inexpensive preliminary exploration of how equilibrium flows respond to marginal cost changes.

We remark in passing that differentiability of a convex function does not in general imply that its gradient is locally Lipschitz. Hence, even though we have characterized the optimal flow as the derivative of a certain convex function (cf.~Theorem \ref{thm:Duality}), $\hbx(\cdot)$ need not be Lipschitz continuous. Strong convexity of the link perturbations is therefore essential for obtaining the global Lipschitz bound in Theorem \ref{thm:CostFlowLipschitzness}. 

\medskip\noindent\textbf{Example \ref{exa:Cost-Functions} (\nameref{exa:Cost-Functions}, Continued)}  
A calculation shows that both the \ref{exa:HalfQuadratic} quadratic and \ref{exa:PiecewiseQuadratic} piecewise quadratic link perturbation functions are $1$--strongly convex, meaning that Theorem \ref{thm:CostFlowLipschitzness} applies with $\mu_e=1$ for all $e\in\cE$. The remaining \ref{exa:NegativeEntropy} entropy and \ref{exa:EntropyLike} entropy-like perturbations are twice differentiable on $\R_{++}$, but their second derivatives converge to zero as $\xi\to\infty$. Hence, these perturbations do not admit a uniform positive lower curvature bound and are therefore not strongly convex.\hfill$\diamondsuit$ 

\medskip\noindent For PURC, \citet{fosgerau_sensitivity_2025} show that the Jacobian $\nabla \widehat{\bx}(\bc)$ is key for such analysis. Using Theorems \ref{thm:Primal-Properties}, \ref{thm:Duality} and \ref{thm:CostFlowAndValueProperties}, we now give a computationally simpler expression than the one that appears in their paper.

\begin{thm}[\textbf{Cost--Flow Differentiability}]
\label{thm:CostFlowDerivative}
Fix $\obc\in\R^\cE$, let $\hbx(\obc)$ denote the unique optimal flow at $\obc$, denote the active and inactive links at $\obc$ by
\begin{equation}\label{eq:ActiveAndInactiveSets}
\cA:=\cS(\obc)
\quad\text{and}\quad \cI:=\cE\setminus\cS(\obc),
\end{equation}
respectively. We refer to the directed graph $(\cV,\cA)$ as the active-link subgraph at $\obc$, and permute coordinates so that the active links $\cA$ appear first in $\cE$ and the inactive links $\cI$ appear second. 
Let $\bA_{\cA}\in\R^{\cV\setminus\{d\}\times \cA}$ be the restriction of $\bA$ to columns indexed by $\cA$, let $H_{\cA}:\R^{\cA}\to(-\infty,\infty]$ be the perturbation function restricted to active links,
\(
H_{\cA}(\bx_{\cA}):=\sum_{e\in \cA} h_e(x_e),
\)
and, when it exists, denote its Hessian at $\hbx_{\cA}(\obc)$ by
\[
\bD_{\cA}:=\nabla^2 H_{\cA}(\hbx_{\cA}(\obc))
=\diag(\{\nabla^2h_e(\hx_e(\obc))\}_{e\in \cA})\in\R^{\cA\times \cA},
\]
where $\hbx(\bc)=(\hbx_{\cA}(\bc),\hbx_{\cI}(\bc))\in\R^{\cA}\times\R^{\cI}$.
Assume the following local regularity at $\obc$:
\begin{enumerate}[label=(\roman*)]
\item\label{enu:SupportStability} \textbf{Support Stability:}
There is an open neighborhood $\cN\subset\R^\cE$ of $\obc$ such that for all $\bc\in\cN$,
the corresponding optimal flow satisfies $\hx_e(\bc)=0$ for all $e\in \cI$.
\item\label{enu:LocalCurvature} \textbf{Local Curvature on Active Links:}
For each $e\in \cA$, there is an open neighborhood $\cN_e\subset\R_{++}$ of $\hx_e(\obc)$ on which $h_e$ is twice continuously differentiable with $\nabla^2h_e(\hx_e(\obc))>0$.
\end{enumerate}
Then $\bc\mapsto\hbx(\bc)$ is continuously differentiable on a (possibly smaller) open
neighborhood of $\obc$. Moreover, in the permuted coordinates, its Jacobian matrix at $\obc$ admits the block representation
\begin{equation}\label{eq:FullJacobianBlock}
\nabla\hbx(\obc)=
\begin{bmatrix}
\nabla_{\bc_{\cA}}\hbx_{\cA}(\obc) & \mathbf 0_{\cA\times \cI}\\
\mathbf 0_{\cI\times \cA} & \mathbf 0_{\cI\times \cI}
\end{bmatrix}
\in\R^{\cE\times\cE},
\end{equation}
where $\bc=(\bc_{\cA},\bc_{\cI})\in\R^{\cA}\times\R^{\cI}$. The upper left $\cA\times \cA$ block is given by
\begin{equation}\label{eq:JacobianBlockNullspaceForm}
\nabla_{\bc_{\cA}}\hbx_{\cA}(\obc)
=
\begin{cases}
\bzero_{\cA\times\cA},
& \text{if } k = 0, \\[1.2ex]
-\bN_{\cA}
\left(
\bN_{\cA}^\top \bD_{\cA}\bN_{\cA}
\right)^{-1}
\bN_{\cA}^\top,
& \text{if } k \geqslant 1.
\end{cases}
\end{equation}
where $k:=\dim\ker(\bA_{\cA})$ is the dimension of the nullspace $\ker(\bA_{\cA})$ of $\bA_{\cA}$ and, in case $k\geqslant1$, $\bN_{\cA}\in\R^{\cA\times k}$ is any matrix having columns forming a basis of this nullspace, so that
$\bA_{\cA}\bN_{\cA}=\bzero$ and $\rank(\bN_{\cA})=k$.
Equivalently, we can express the upper left $\cA\times\cA$ block by
\begin{equation}\label{eq:JacobianBlockLaplacianForm}
\nabla_{\bc_{\cA}}\hbx_{\cA}(\obc)
=
-\left[\bD_{\cA}^{-1}
-\bD_{\cA}^{-1}\bA_{\cA}^\top
\bL_{\cA}^{\dagger}
\bA_{\cA}\bD_{\cA}^{-1}\right],
\end{equation}
with $\bL_{\cA}:=\bA_{\cA}\bD_{\cA}^{-1}\bA_{\cA}^\top\in\R^{\cV\backslash\{d\}\times\cV\backslash\{d\}}$ 
denoting the inverse-curvature weighted Laplacian associated with $\bA_{\cA}$, and $\left(\cdot\right)^\dagger$ denotes Moore--Penrose pseudoinversion.
\end{thm}

\noindent The Jacobian can be applied in the context of sensitivity analysis of the equilibrium assignment using the PURC model, replacing the expression in \citet[Theorem 1]{fosgerau_sensitivity_2025} by the new, computationally simpler expression. 

\begin{rem}[\textbf{Nullspace Basis Invariance}]
The expression \eqref{eq:JacobianBlockNullspaceForm} in the case of a nontrivial nullspace $\ker(\bA_{\cA})$, so that $k\geqslant1$, does not depend on the choice of basis. Indeed, if $\bN_{\cA}\in\R^{\cA\times k}$ and $\widetilde\bN_{\cA}\in\R^{\cA\times k}$ are matrices whose columns form bases of the nontrivial nullspace $\ker(\bA_{\cA})$, then there is an invertible matrix $\bQ_{\cA}\in\R^{k\times k}$, such that $\widetilde\bN_{\cA}=\bN_{\cA}\bQ_{\cA}$. Since
\(
\widetilde\bN_{\cA}(\widetilde\bN_{\cA}^\top\bD_{\cA}\widetilde\bN_{\cA})^{-1}\widetilde\bN_{\cA}^\top
=
\bN_{\cA}(\bN_{\cA}^\top\bD_{\cA}\bN_{\cA})^{-1}\bN_{\cA}^\top,
\)
the invariance follows.\hfill$\diamondsuit$
\end{rem} 

\begin{rem}[\textbf{Nullspace Basis Computation}]
A basis can be obtained, for example, from a singular value decomposition (SVD) of $\bA_{\cA}$. Writing
\(
\bA_{\cA} = \bU \boldsymbol{\Sigma} \bV^\top,
\)
the columns of $\bV$ corresponding to the zero singular values form an orthonormal basis of $\ker(\bA_{\cA})$ and may be taken as $\bN_{\cA}$. Equivalently, one may compute a rank-revealing QR factorization or use any standard numerical nullspace routine.\hfill$\diamondsuit$
\end{rem} 

\begin{rem}[\textbf{Computation via Cycle-Space Adjustment}]
Suppose $k=\dim\ker(\bA_{\cA})\geqslant 1$ and let 
$\bN_{\cA}\in\R^{\cA\times k}$ have columns forming a basis of 
$\ker(\bA_{\cA})$. For a small cost perturbation $\mathrm d\bc\in\R^\cE$ from $\obc$, write $\mathrm d\bc_{\cA}$ for its restriction to active links. Under the assumptions of Theorem~\ref{thm:CostFlowDerivative}, the inactive links remain unchanged, $\mathrm d\hbx_{\cI}=\mathbf 0$, and the active-link adjustment can be broken down into three steps:
\begin{enumerate}[label=\emph{Step \arabic*.}, leftmargin=*]
\item \emph{Unconstrained adjustment.}  
Ignoring flow conservation, each active link would adjust according to its local curvature,
\[
\mathrm d\tbx_{\cA}
:=
-\bD_{\cA}^{-1}\,\mathrm d\bc_{\cA},
\]
which can be interpreted as the link-wise best response to the cost perturbation.
\item \emph{Restriction to feasible directions.}  
Flow conservation requires that admissible adjustments lie in 
$\ker(\bA_{\cA})$, i.e., along cycles of the active-link subgraph.
Project the unconstrained response onto this cycle space in the 
$\bD_{\cA}$-metric by solving
\[
(
\bN_{\cA}^\top \bD_{\cA}\bN_{\cA}
)\bw
=
\bN_{\cA}^\top \bD_{\cA}\mathrm d\tbx_{\cA}
\]
for the weights $\bw\in\R^{k}$.

\item \emph{Cycle-space correction.}  
The feasible first-order adjustment is then
\[
\mathrm d\hbx_{\cA}
=
\bN_{\cA}\bw.
\]
\end{enumerate}
This procedure makes transparent that flow sensitivity operates through the cycle structure of the active-link subgraph: cost perturbations first generate local link responses, which are then filtered through circulations that preserve node balance.

While the three-step construction clarifies the geometric structure of the result, implementation may proceed directly via \eqref{eq:JacobianBlockNullspaceForm}.\hfill$\diamondsuit$
\end{rem} 

\begin{rem}[\textbf{Computation via Laplacian Adjustment}]
Let $\mathrm d\bc\in\R^\cE$ be a small cost perturbation from $\obc$ and write 
$\mathrm d\bc_{\cA}$ for its restriction to active links. Under the assumptions of 
Theorem~\ref{thm:CostFlowDerivative}, the inactive links remain unchanged, 
$\mathrm d\hbx_{\cI}=\mathbf 0$, and the active-link adjustment can be broken down into three steps:

\begin{enumerate}[label=\emph{Step \arabic*.}, leftmargin=*]

\item \emph{Unconstrained adjustment.}  
Ignoring flow conservation, each active link would adjust according to its local curvature,
\[
\mathrm d\tbx_{\cA}
:=
-\bD_{\cA}^{-1}\,\mathrm d\bc_{\cA},
\]
which can again be interpreted as the link-wise best response to the cost perturbation.

\item \emph{Node-potential adjustment.}  
The unconstrained response generally violates node balance, since 
$\bA_{\cA}\mathrm d\tbx_{\cA}\neq \mathbf 0$.  
To restore feasibility, determine potentials 
$\bu\in\R^{\cV\setminus\{d\}}$ solving
\[
\bL_{\cA}\,\bu
=
\bA_{\cA}\mathrm d\tbx_{\cA},
\]
where we recall $\bL_{\cA}:=\bA_{\cA}\bD_{\cA}^{-1}\bA_{\cA}^\top$ is the inverse-curvature weighted Laplacian.  
If $\bL_{\cA}$ is singular, take the minimum-norm solution 
$\bu:=\bL_{\cA}^{\dagger}\bA_{\cA}\mathrm d\tbx_{\cA}$. 

\item \emph{Feasibility correction.}  
Subtract the potential-induced edge adjustment to obtain a circulation,
\[
\mathrm d\hbx_{\cA}
=
\mathrm d\tbx_{\cA}
-
\bD_{\cA}^{-1}\bA_{\cA}^\top \bu.
\]

\end{enumerate}
This formulation emphasizes the node-based interpretation of sensitivity: cost perturbations first induce local link responses, which are then adjusted via potentials to restore flow conservation. The Laplacian solve enforces node balance by removing the component that would violate feasibility.

As above, the three-step construction is provided to clarify the geometric and economic structure of the result. In practice, one may compute $\nabla\hbx(\obc)$ directly from 
\eqref{eq:JacobianBlockNullspaceForm} or \eqref{eq:JacobianBlockLaplacianForm}. 
In large sparse networks, the Laplacian formulation is typically computationally preferable, as it requires solving a system of dimension $|\cV|-1$ rather than constructing a nullspace basis of $\bA_{\cA}$.\hfill$\diamondsuit$
\end{rem}
 
\begin{rem}[\textbf{Reduction and Simplified Jacobian under Single O--D Demand}]\label{rem:SimplificationViaReducedLaplacian}
Let 
\[
\cV_{\cA}
:=
\left\{
v\in \cV\setminus\{d\}
\middle|
v \text{ is incident to some } e\in\cA
\right\}
\]
denote the set of non-destination nodes incident to at least one active link, and let
\(
\bA_{\cV_{\cA},\cA}\in\R^{\cV_{\cA}\times\cA}
\)
be the submatrix of $\bA$ obtained by restricting rows to $\cV_{\cA}$ and columns to $\cA$ (after the destination row has been removed in the construction of $\bA$). Define the \emph{reduced inverse-curvature weighted Laplacian}
\(
\bL_{\cA}^{(R)}
:=
\bA_{\cV_{\cA},\cA}\,
\bD_{\cA}^{-1}\,
\bA_{\cV_{\cA},\cA}^{\top}
\in\R^{\cV_{\cA}\times\cV_{\cA}}.
\)
Then the potential equation 
$\bL_{\cA}\bu=\bA_{\cA}\mathrm d\tbx_{\cA}$ 
reduces to
\[
\bL_{\cA}^{(R)}\,\bu^{(R)}
=
\bA_{\cV_{\cA},\cA}\,\mathrm d\tbx_{\cA},
\qquad
\bu^{(R)}\in\R^{\cV_{\cA}},
\]
with $\bu_v=0$ for nodes outside $\cV_{\cA}$ in the minimum-norm solution.

In the present single origin--destination setting, flow conservation implies that every active link lies on an $o$--$d$ path. Hence the underlying undirected graph of the active-link subgraph restricted to $\cV_{\cA}\cup\{d\}$ is connected, implying that $\bL_{\cA}^{(R)}$ is invertible. Consequently, the active-link Jacobian admits the simplified representation
\begin{equation}\label{eq:JacobianBlockReducedLaplacianForm}
\nabla_{\bc_{\cA}}\hbx_{\cA}(\obc)
=
-\big[
\bD_{\cA}^{-1}
-
\bD_{\cA}^{-1}
\bA_{\cV_{\cA},\cA}^{\top}
\big(
\bA_{\cV_{\cA},\cA}\,
\bD_{\cA}^{-1}\,
\bA_{\cV_{\cA},\cA}^{\top}
\big)^{-1}
\bA_{\cV_{\cA},\cA}\,
\bD_{\cA}^{-1}
\big],
\end{equation}
and the full Jacobian $\nabla\hbx(\obc)$ is obtained by embedding this block into \eqref{eq:FullJacobianBlock}.\hfill$\diamondsuit$
\end{rem} 

\noindent We conclude this section by combining the findings of Theorems \ref{thm:CostFlowAndValueProperties} and \ref{thm:CostFlowDerivative} and Remark \ref{rem:SimplificationViaReducedLaplacian} into the following corollary.

\begin{cor}[\textbf{Cost--Flow Differentiability under Global Twice Continuous Differentiability and Curvature}]
\label{cor:CostFlowDifferentiabilityUnderC2AndCurvature}
Assume that $\obc\in\cC_0$ and that each link perturbation
$h_e$, $e\in\cE$, is twice continuously differentiable on $\operatorname{int}(\dom\,h_e)=\R_{++}$ and satisfies
$\nabla^2 h_e(\xi)>0$ for all $\xi\in\R_{++}$.
Then the map $\bc\mapsto\hbx(\bc)$ is continuously differentiable on an open neighborhood of $\obc$, its Jacobian $\nabla\hbx(\obc)$ admits the block form
\eqref{eq:FullJacobianBlock}, and the upper-left active-link block $\nabla_{\bc_\cA}\hbx_\cA(\obc)$ (with $\cA=\cS(\obc)$) is given by either of the equivalent
expressions \eqref{eq:JacobianBlockNullspaceForm}, \eqref{eq:JacobianBlockLaplacianForm} or \eqref{eq:JacobianBlockReducedLaplacianForm}.
\end{cor}

\noindent Corollary~\ref{cor:CostFlowDifferentiabilityUnderC2AndCurvature} converts the local hypotheses of Theorem \ref{thm:CostFlowDerivative} into two conceptually simpler sufficient conditions. First, membership of $\obc$ in the open dense set $\cC_0$ guarantees support stability: by Theorem~\ref{thm:CostFlowAndValueProperties}\ref{enu:CostSupportConstancy}, the set of active (and therefore inactive) links is locally constant
at such cost vectors. Second, imposing twice continuous differentiability and strict curvature of each link perturbation $h_e$ on all of $\R_{++}(=\operatorname{int}(\dom\,h_e))$ implies the required local curvature at the positive link flows automatically. The corollary shows that differentiability and the explicit Jacobian representation are not exceptional: these findings hold on a topologically generic subset of $\R^\cE$ under standard smoothness/curvature assumptions on the link perturbations.


\begin{example}[\textbf{Flow Jacobian in a Two-Node, Two-Link Network}]
\label{exa:TwoNodeTwoLinkSensitivity}
Continue with the problem from Example~\ref{exa:TwoNodeTwoLinkDuality}, where the optimal flow takes the form in \eqref{eq:TwoNodeTwoLinkOptimalFlow}. Here $\bc\mapsto\hbx(\bc)$ is differentiable at every point. Direct calculation produces a Jacobian of the form
\begin{equation}\label{eq:AnalyticExampleFlowJacobianDirectCalculation}
\nabla \hbx(\bc)=
-\frac{\e^{-c_{e_1}-c_{e_2}-2}}{\sqrt{1+4\e^{-c_{e_1}-c_{e_2}-2}}}
\begin{bmatrix}1&1\\[2pt]1&1\end{bmatrix},\quad\bc\in\R^\cE.
\end{equation}
The Jacobian shows that increasing the cost of travel on either link reduces the optimal flow on both links.
We next arrive at the same expression from the routes provided by Theorem \ref{thm:CostFlowDerivative}. 

To this end, note that the negative entropy perturbation cost is twice continuously differentiable on all of $\R_{++}$ with $\nabla^2 h_e(\xi)=1/\xi$. Since $\bc\mapsto\hbx(\bc)$ maps into $\R^{\cE}_{++}$, both links are always active, i.e., $\cA=\{e_1,e_2\}$. We therefore have $\bD_{\cA}=\nabla^2 H(\hbx(\obc))=\diag(1/\hx_{e_1}(\obc),1/\hx_{e_2}(\obc))$, implying $\bD_{\cA}^{-1}=\diag(\hx_{e_1}(\obc),\hx_{e_2}(\obc))$. 

To evaluate the expression in \eqref{eq:JacobianBlockNullspaceForm}, observe that the (restricted) incidence matrix $\bA_{\cA}=\bA=[-1,1]$ has the nullspace $\ker(\bA_{\cA})=\ker(\bA)=\{\tau\bone_{\cE}|\tau\in\R\}$, so that we can take $\bN_{\cA}$ to be the $\R^{\cE\times 1}$ matrix comprised by the single column $\bone_\cE\in\R^\cE$. The reciprocal of $\bN_{\cA}^\top\bD_{\cA}^{-1}\bN_{\cA}=\hx_{e_1}(\obc)+\hx_{e_2}(\obc)$ becomes the leading scalar in \eqref{eq:AnalyticExampleFlowJacobianDirectCalculation} and $\bN_{\cA}\bN_{\cA}^\top$ the matrix of ones.

To evaluate the expression in \eqref{eq:JacobianBlockLaplacianForm}, observe that the weighted Laplacian $\bL_{\cA}=\bA_{\cA}\bD_\cA^{-1}\bA_{\cA}^\top$ reduces to $\hx_{e_1}(\obc)+\hx_{e_2}(\obc)>0$, implying that $\bL_{\cA}$ is invertible with $\bL_{\cA}^{-1}=1/[\hx_{e_1}(\obc)+\hx_{e_2}(\obc)]$. Inserting the $\bD_{\cA},\bA_{\cA}$ and $\bL_{\cA}^{-1}$, the right-hand side of \eqref{eq:JacobianBlockLaplacianForm} becomes
\[
-\frac{\hx_{e_1}(\bc_0) \hx_{e_2}(\bc_0)}{\hx_{e_1}(\bc_0)+\hx_{e_2}(\bc_0)}
\begin{bmatrix}1&1\\[2pt]1&1\end{bmatrix}.
\]
Inserting the optimal flow \eqref{eq:TwoNodeTwoLinkOptimalFlow}, we again reproduce \eqref{eq:AnalyticExampleFlowJacobianDirectCalculation}.\hfill$\diamondsuit$
\end{example}

\section{Numerical Example}\label{sec:example}

We illustrate the PURC model and our main results using a simple network and the half-quadratic link perturbation from Example \ref{exa:Cost-Functions}. Figure \ref{fig:illustrative-network} shows a strongly connected network with four nodes and six directed links. The set of nodes is $\cV=\{1,2,3,4\}$, links \( \cE=\{(1,2),(1,3),(2,3),(2,4),(3,4),(4,1)\} \) are identified by ordered pairs of nodes, and link costs are denoted $c_{i,j},(i,j)\in\cE$.
\begin{figure}[htbp]
  \centering
  \includegraphics[width=0.7\linewidth]{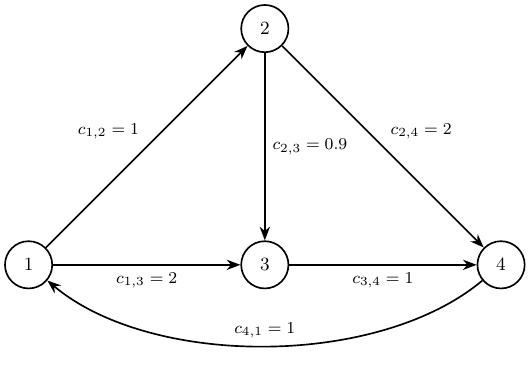}
  \caption{Example network with four nodes and six directed links. The traveler moves from node 1 (origin) to node 4 (destination).}
  \label{fig:illustrative-network}
\end{figure}
The (complete) incidence matrix, (complete) demand vector, and cost vector are, respectively,
\[
\overline{\bA}=\begin{bmatrix*}[r]
-1 & -1 &  0 &  0 &  0 &  1\\
 1 &  0 & -1 & -1 &  0 &  0\\
 0 &  1 &  1 &  0 & -1 &  0\\
 0 &  0 &  0 &  1 &  1 & -1
\end{bmatrix*},\quad
\overline{\bb}=\begin{bmatrix*}[r]
-1\\0\\0\\1\end{bmatrix*},\quad
\text{ and }\quad
\bc=[1,2,0.9,2,1,1]^\top.
\]
We use the half-quadratic perturbation in Example \ref{exa:Cost-Functions}(d) for all links, leading to a (common) conjugate $h_e^*(\eta)=\tfrac{1}{2}\max\{0,\eta\}^2$ with gradient $\nabla h_e^*(\eta)=\max\{0,\eta\}$. Dropping the destination row and entry to get $\bA$ and $\bb$, respectively, the traveler's problem is
\begin{align*}
\Minimize&\quad\sum_{e\in \cE} \left( c_e x_e + \tfrac{1}{2} x_e^2 \right)\text{ over flows }\bx\in\R^\cE\\
\text{satisfying}&\quad\bx\geqslant\bzero\text{ and }\bA\bx=\bb,
\end{align*}
which is a convex quadratic program in six dimensions.

Instead of solving the traveler's problem directly, we use the dual formulation (Theorem \ref{thm:PstarIsPDual}) and
\[
\Maximize\;g(\bu)=\langle\bb,\bu\rangle-\tfrac{1}{2}\sum_{e\in\cE}\max\{0,\langle\ba_e,\bu\rangle-c_e\}^2\text{ over potentials }\bu\in\R^{\cV\setminus\{d\}}.
\]
The traveler's dual problem is unconstrained, differentiable, and concave in three dimensions. We solve it via a simple gradient ascent scheme:
\begin{itemize}
    \item From the current guess $\bu$, compute the inner products $t_e(\bu):=\langle\ba_e,\bu\rangle,e\in\cE$.
    \item Compute the flow $\bx^\ast(\bu)\in\R^\cE$ defined by $x_e^\ast(\bu):=\nabla h_e^\ast(t_e(\bu)-c_e)$, which is here\[x_e^\ast(\bu)=\max\{0,t_e(\bu)-c_e\}.\]
    \item Compute the gradient of the dual function, $\nabla g(\bu)=\bb - \bA\bx^\ast(\bu)$.
    \item Update $\bu\leftarrow\bu + \alpha \nabla g(\bu)$ and start over.
\end{itemize}
\noindent Starting from $\bu=\bzero$, using the step size $\alpha=0.05$, and iterating until $\|\nabla g(\bu)\|$ is below a chosen tolerance, we obtain an optimal potential 
and the optimal flow, respectively,
\[
\widehat{\bu}=[-4,-2.475,-1.525]^\top\quad\text{and}\quad
\widehat{\bx}=[0.525,0.475,0.050,0.475,0.525,0]^\top.
\]
The potential at a node reflects its generalized marginal cost to reach the destination. The optimal flow sends mass $0.475$ along paths $1\to2\to4$ and $1\to3\to4$, and $0.05$ along $1\to2\to3\to4$.

At $\bc$, link $(4,1)$ is strictly inactive since 
$\langle \ba_{4,1},\widehat{\bu}\rangle - c_{4,1} = -5 < 0$, while all other links are strictly active. Since the half-quadratic perturbation is twice continuously differentiable on $\R_{++}$ with $\nabla^2 h_e \equiv 1$, the assumptions of Theorem~\ref{thm:CostFlowDerivative} are satisfied. Applying Theorem \ref{thm:CostFlowDerivative}, formulas \eqref{eq:JacobianBlockNullspaceForm},  \eqref{eq:JacobianBlockLaplacianForm} and \eqref{eq:JacobianBlockReducedLaplacianForm} all lead to
\[
\nabla \widehat{\bx}(\bc) = 
\begin{bmatrix*}[r]
-0.375 &  0.375 & -0.250 & -0.125 &  0.125 &  0 \\
 0.375 & -0.375 &  0.250 &  0.125 & -0.125 &  0 \\
-0.250 &  0.250 & -0.500 &  0.250 & -0.250 &  0 \\
-0.125 &  0.125 &  0.250 & -0.375 &  0.375 &  0 \\
 0.125 & -0.125 & -0.250 &  0.375 & -0.375 &  0 \\
 0 &  0 &  0 &  0 &  0 &  0
\end{bmatrix*}
\]
Thus, for example, increasing the cost of link $(1,2)$, we find flow decreases on links $(1,2),(2,3)$ and $(2,4)$ and increases on links $(1,3)$ and $(3,4)$, while it is unchanged on link $(4,1)$.

\section{Concluding Remarks}\label{sec:conclusion}
We have developed a rigorous convex-analytic foundation for the perturbed utility route choice (PURC) model. A key result is that the traveler’s constrained, possibly non-smooth utility maximization problem admits a computationally tractable dual formulation. The dual problem is unconstrained, concave, differentiable, and almost everywhere twice differentiable in the space of node potentials, even when the primal link perturbations are non-differentiable. The unique optimal flow can be easily recovered from any dual solution along with the Jacobian of the optimal flow with respect to link costs.

More broadly, the PURC model is analogous to a resistive network: potentials correspond to voltages, flows to currents, and link perturbations to generalized resistance.

Future research could build on the tractability of the dual problem to develop faster algorithms for parameter estimation and equilibrium assignment. Another direction is to explore connections to monotropic programming and extend beyond the PURC framework. Finally, it would be valuable to deepen the interpretation of link perturbations as information costs by developing a more explicit theory of information acquisition and processing.

\bibliographystyle{ecta}
\bibliography{purc_network_flows,MF}

\begin{thebibliography}{44}
\newcommand{\enquote}[1]{``#1''}
\expandafter\ifx\csname natexlab\endcsname\relax\def\natexlab#1{#1}\fi

\bibitem[\protect\citeauthoryear{Aleksandrov}{Aleksandrov}{1939}]{Alexandroff1939}
\textsc{Aleksandrov, A.~D.} (1939): \enquote{Almost everywhere existence of the second differential of a convex function and some properties of convex surfaces connected with it,} \emph{Uchenye Zapiski Leningrad Gos. Univ., Math. Ser}, 6, 3--35.

\bibitem[\protect\citeauthoryear{Allen and Rehbeck}{Allen and Rehbeck}{2019}]{allen2019identification}
\textsc{Allen, R. and J.~Rehbeck} (2019): \enquote{Identification with additively separable heterogeneity,} \emph{Econometrica}, 87, 1021--1054.

\bibitem[\protect\citeauthoryear{Allen and Rehbeck}{Allen and Rehbeck}{2020}]{allen2020hicksian}
---\hspace{-.1pt}---\hspace{-.1pt}--- (2020): \enquote{Hicksian complementarity and perturbed utility models,} \emph{Economic Theory Bulletin}, 8, 245--261.

\bibitem[\protect\citeauthoryear{Baillon and Cominetti}{Baillon and Cominetti}{2008}]{baillon_markovian_2008}
\textsc{Baillon, J.~B. and R.~Cominetti} (2008): \enquote{Markovian traffic equilibrium,} \emph{Mathematical Programming}, 111, 33--56.

\bibitem[\protect\citeauthoryear{Bell}{Bell}{1995}]{bell_alternatives_1995}
\textsc{Bell, M.~G.} (1995): \enquote{Alternatives to {Dial}'s logit assignment algorithm,} \emph{Transportation Research Part B}, 29, 287--295.

\bibitem[\protect\citeauthoryear{Ben-Akiva and Bowman}{Ben-Akiva and Bowman}{1998}]{ben-akiva_activity_1998}
\textsc{Ben-Akiva, M.~E. and J.~L. Bowman} (1998): \enquote{Activity {Based} {Travel} {Demand} {Model} {Systems},} in \emph{Equilibrium and {Advanced} {Transportation} {Modelling}}, ed. by P.~Marcotte and S.~Nguyen, Boston, MA: Springer US, 27--46.

\bibitem[\protect\citeauthoryear{Bertsekas}{Bertsekas}{1998}]{bertsekas1998network}
\textsc{Bertsekas, D.} (1998): \emph{Network Optimization: Continuous and Discrete Models}, vol.~8, Athena Scientific.

\bibitem[\protect\citeauthoryear{Bertsekas}{Bertsekas}{2009}]{bertsekas2009convex}
---\hspace{-.1pt}---\hspace{-.1pt}--- (2009): \emph{Convex Optimization Theory}, vol.~1, Athena Scientific.

\bibitem[\protect\citeauthoryear{Blondel, Martins, and Niculae}{Blondel et~al.}{2020}]{blondel_learning_2020}
\textsc{Blondel, M., A.~F.~T. Martins, and V.~Niculae} (2020): \enquote{Learning with {Fenchel}-{Young} losses,} \emph{Journal of Machine Learning Research}, 21, 1--69.

\bibitem[\protect\citeauthoryear{Bregman}{Bregman}{1967}]{bregman_relaxation_1967}
\textsc{Bregman, L.} (1967): \enquote{The relaxation method of finding the common point of convex sets and its application to the solution of problems in convex programming,} \emph{USSR Computational Mathematics and Mathematical Physics}, 7, 200--217.

\bibitem[\protect\citeauthoryear{de~Cea and Fernández}{de~Cea and Fernández}{1993}]{de_cea_transit_1993}
\textsc{de~Cea, J. and E.~Fernández} (1993): \enquote{Transit {Assignment} for {Congested} {Public} {Transport} {Systems}: {An} {Equilibrium} {Model},} \emph{Transportation Science}, 27, 133--147, publisher: INFORMS.

\bibitem[\protect\citeauthoryear{Dial}{Dial}{1971}]{Dial1971}
\textsc{Dial, R.~B.} (1971): \enquote{A probabilistic multipath traffic assignment algorithm which obviates path enumeration,} \emph{Transportation Research}, 5, 83--111.

\bibitem[\protect\citeauthoryear{Fosgerau, Frejinger, and Karlstrom}{Fosgerau et~al.}{2013}]{fosgerau_link_2013}
\textsc{Fosgerau, M., E.~Frejinger, and A.~Karlstrom} (2013): \enquote{A link based network route choice model with unrestricted choice set,} \emph{Transportation Research Part B: Methodological}, 56, 70--80.

\bibitem[\protect\citeauthoryear{Fosgerau, {\L}ukawska, Paulsen, and Rasmussen}{Fosgerau et~al.}{2023}]{fosgerau2023bikeability}
\textsc{Fosgerau, M., M.~{\L}ukawska, M.~Paulsen, and T.~K. Rasmussen} (2023): \enquote{Bikeability and the induced demand for cycling,} \emph{Proceedings of the National Academy of Sciences}, 120, e2220515120.

\bibitem[\protect\citeauthoryear{Fosgerau and McFadden}{Fosgerau and McFadden}{2012}]{McFadden2012}
\textsc{Fosgerau, M. and D.~L. McFadden} (2012): \enquote{A theory of the perturbed consumer with general budgets,} \emph{NBER Working Paper}, 1--27, publisher: National Bureau of Economic Research.

\bibitem[\protect\citeauthoryear{Fosgerau, Monardo, and de~Palma}{Fosgerau et~al.}{2024{\natexlab{a}}}]{fosgerau_inverse_2024}
\textsc{Fosgerau, M., J.~Monardo, and A.~de~Palma} (2024{\natexlab{a}}): \enquote{The {Inverse} {Product} {Differentiation} {Logit} {Model},} \emph{American Economic Journal: Microeconomics}, 16, 329--70.

\bibitem[\protect\citeauthoryear{Fosgerau, Nielsen, Paulsen, Rasmussen, and Yao}{Fosgerau et~al.}{2025}]{fosgerau_sensitivity_2025}
\textsc{Fosgerau, M., N.~Nielsen, M.~Paulsen, T.~K. Rasmussen, and R.~Yao} (2025): \enquote{Sensitivity {Analysis} of the {Perturbed} {Utility} {Stochastic} {Traffic} {Equilibrium},} .

\bibitem[\protect\citeauthoryear{Fosgerau, Paulsen, and Rasmussen}{Fosgerau et~al.}{2022}]{fosgerau2022perturbed}
\textsc{Fosgerau, M., M.~Paulsen, and T.~K. Rasmussen} (2022): \enquote{A perturbed utility route choice model,} \emph{Transportation Research Part C: Emerging Technologies}, 136, 103514.

\bibitem[\protect\citeauthoryear{Fosgerau, Yao, Nielsen, Paulsen, and Rasmussen}{Fosgerau et~al.}{2024{\natexlab{b}}}]{fosgerau_estimating_2024}
\textsc{Fosgerau, M., R.~Yao, N.~Nielsen, M.~Paulsen, and T.~K. Rasmussen} (2024{\natexlab{b}}): \enquote{Estimating the perturbed utility route choice model with trip-level data,} .

\bibitem[\protect\citeauthoryear{Fudenberg, Iijima, and Strzalecki}{Fudenberg et~al.}{2015}]{Fudenberg2015}
\textsc{Fudenberg, D., R.~Iijima, and T.~Strzalecki} (2015): \enquote{Stochastic {Choice} and {Revealed} {Perturbed} {Utility},} \emph{Econometrica}, 83, 2371--2409.

\bibitem[\protect\citeauthoryear{Hiriart-Urruty and Lemar{\'e}chal}{Hiriart-Urruty and Lemar{\'e}chal}{2004}]{hiriart2004fundamentals}
\textsc{Hiriart-Urruty, J.-B. and C.~Lemar{\'e}chal} (2004): \emph{Fundamentals of Convex Analysis}, Springer Science \& Business Media.

\bibitem[\protect\citeauthoryear{Hofbauer and Sandholm}{Hofbauer and Sandholm}{2002}]{Hofbauer2002}
\textsc{Hofbauer, J. and W.~H. Sandholm} (2002): \enquote{On the global convergence of stochastic fictitious play,} \emph{Econometrica}, 70, 2265--2294.

\bibitem[\protect\citeauthoryear{Kechris}{Kechris}{1995}]{kechris_classical_1995}
\textsc{Kechris, A.~S.} (1995): \enquote{Classical descriptive set theory,} in \emph{Classical descriptive set theory}, New York, NY: Springer New York, Graduate texts in mathematics, 156, 1st ed. 1995. ed.

\bibitem[\protect\citeauthoryear{Mai}{Mai}{2016}]{mai_method_2016}
\textsc{Mai, T.} (2016): \enquote{A method of integrating correlation structures for a generalized recursive route choice model,} \emph{Transportation Research Part B: Methodological}, 93, 146--161.

\bibitem[\protect\citeauthoryear{Mai, Fosgerau, and Frejinger}{Mai et~al.}{2015{\natexlab{a}}}]{Mai2015}
\textsc{Mai, T., M.~Fosgerau, and E.~Frejinger} (2015{\natexlab{a}}): \enquote{A nested recursive logit model for route choice analysis,} \emph{Transportation Research Part B: Methodological}, 75, 100--112.

\bibitem[\protect\citeauthoryear{Mai, Frejinger, and Bastin}{Mai et~al.}{2015{\natexlab{b}}}]{Mai2015a}
\textsc{Mai, T., E.~Frejinger, and F.~Bastin} (2015{\natexlab{b}}): \enquote{A {Dynamic} {Programming} {Approach} for {Quickly} {Estimating} {Large} {Scale} {MEV} {Models},} \emph{Transportation Research Part B: Methodological}.

\bibitem[\protect\citeauthoryear{McFadden}{McFadden}{1981}]{mcfadden_econometric_1981}
\textsc{McFadden, D.} (1981): \enquote{Econometric {Models} of {Probabilistic} {Choice},} in \emph{Structural {Analysis} of {Discrete} {Data} with {Econometric} {Applications}}, ed. by C.~Manski and D.~McFadden, Cambridge, MA, USA: MIT Press, 198--272, iSSN: 00219398.

\bibitem[\protect\citeauthoryear{Nesterov}{Nesterov}{2018}]{Nesterov2018}
\textsc{Nesterov, Y.} (2018): \emph{Lectures on convex optimization}, vol. 137, Springer.

\bibitem[\protect\citeauthoryear{Oyama, Hara, and Akamatsu}{Oyama et~al.}{2022}]{oyama_markovian_2022}
\textsc{Oyama, Y., Y.~Hara, and T.~Akamatsu} (2022): \enquote{Markovian traffic equilibrium assignment based on network generalized extreme value model,} \emph{Transportation Research Part B: Methodological}, 155, 135--159.

\bibitem[\protect\citeauthoryear{Patriksson and Rockafellar}{Patriksson and Rockafellar}{2003}]{patriksson_sensitivity_2003}
\textsc{Patriksson, M. and R.~T. Rockafellar} (2003): \enquote{Sensitivity {Analysis} of {Aggregated} {Variational} {Inequality} {Problems}, with {Application} to {Traffic} {Equilibria},} \emph{Transportation Science}, 37, 56--68, publisher: INFORMS.

\bibitem[\protect\citeauthoryear{Prato}{Prato}{2009}]{prato_route_2009}
\textsc{Prato, C.~G.} (2009): \enquote{Route choice modeling: {Past}, present and future research directions,} \emph{Journal of Choice Modelling}, 2, 65--100, iSBN: 1755-5345.

\bibitem[\protect\citeauthoryear{Ran and Boyce}{Ran and Boyce}{1996}]{ran_modeling_1996}
\textsc{Ran, B. and D.~Boyce} (1996): \emph{Modeling {Dynamic} {Transportation} {Networks}}, Berlin, Heidelberg: Springer.

\bibitem[\protect\citeauthoryear{Roberts and Kroese}{Roberts and Kroese}{2007}]{roberts_estimating_2007}
\textsc{Roberts, B. and D.~Kroese} (2007): \enquote{Estimating the {Number} of s-t {Paths} in a {Graph},} \emph{Journal of Graph Algorithms and Applications}, 11, 195--214.

\bibitem[\protect\citeauthoryear{Rockafellar}{Rockafellar}{1970}]{rockafellar1970convex}
\textsc{Rockafellar, R.~T.} (1970): \emph{Convex Analysis}, vol.~28, Princeton University Press.

\bibitem[\protect\citeauthoryear{Rockafellar}{Rockafellar}{1984}]{rockafellar1984network}
---\hspace{-.1pt}---\hspace{-.1pt}--- (1984): \emph{Network Flows and Monotropic Optimization}, John Wiley \& Sons.

\bibitem[\protect\citeauthoryear{Sheffi and Daganzo}{Sheffi and Daganzo}{1980}]{sheffi_computation_1980}
\textsc{Sheffi, Y. and C.~F. Daganzo} (1980): \enquote{Computation of {Equilibrium} {Over} {Transportation} {Networks}: {The} {Case} of {Disaggregate} {Demand} {Models},} \emph{Transportation Science}, 14, 155--173, publisher: INFORMS.

\bibitem[\protect\citeauthoryear{Shen, Zhang, and Akamatsu}{Shen et~al.}{1996}]{shen_cyclic_1996}
\textsc{Shen, W., H.~M. Zhang, and T.~Akamatsu} (1996): \enquote{Cyclic flows, {Markov} process and stochastic traffic assignment,} \emph{Transportation Research Record}, 30, 1--34, iSBN: 0191-2615.

\bibitem[\protect\citeauthoryear{Small and Verhoef}{Small and Verhoef}{2007}]{Small1992}
\textsc{Small, K.~A. and E.~T. Verhoef} (2007): \emph{Urban transportation economics}, London and New York: Harwood Academic Publishers.

\bibitem[\protect\citeauthoryear{S{\o}rensen and Fosgerau}{S{\o}rensen and Fosgerau}{2022}]{sorensen2022mcfadden}
\textsc{S{\o}rensen, J. R.-V. and M.~Fosgerau} (2022): \enquote{How McFadden met Rockafellar and learned to do more with less,} \emph{Journal of Mathematical Economics}, 100, 102629.

\bibitem[\protect\citeauthoryear{Tobin and Friesz}{Tobin and Friesz}{1988}]{tobin_sensitivity_1988}
\textsc{Tobin, R.~L. and T.~L. Friesz} (1988): \enquote{Sensitivity {Analysis} for {Equilibrium} {Network} {Flow},} \emph{Transportation Science}, 22, 242--250, publisher: INFORMS.

\bibitem[\protect\citeauthoryear{Yang and Bell}{Yang and Bell}{1997}]{yang_traffic_1997}
\textsc{Yang, H. and M.~G.~H. Bell} (1997): \enquote{Traffic restraint, road pricing and network equilibrium,} \emph{Transportation Research Part B: Methodological}, 31, 303--314.

\bibitem[\protect\citeauthoryear{Yao, Fosgerau, Paulsen, and Rasmussen}{Yao et~al.}{2024}]{yao_perturbed_2024}
\textsc{Yao, R., M.~Fosgerau, M.~Paulsen, and T.~Rasmussen} (2024): \enquote{Perturbed utility stochastic traffic assignment,} \emph{Transportation Science}, 58.

\bibitem[\protect\citeauthoryear{Zalinescu}{Zalinescu}{2002}]{zalinescu2002convex}
\textsc{Zalinescu, C.} (2002): \emph{Convex analysis in general vector spaces}, World scientific.

\bibitem[\protect\citeauthoryear{Zimmermann and Frejinger}{Zimmermann and Frejinger}{2020}]{zimmermann_tutorial_2020}
\textsc{Zimmermann, M. and E.~Frejinger} (2020): \enquote{A tutorial on recursive models for analyzing and predicting path choice behavior,} \emph{EURO Journal on Transportation and Logistics}, 9, 100004, arXiv:1905.00883 [stat].

\end{thebibliography}

\newpage{}

\appendix

\part*{Appendices}

\section{Convex Analysis Preliminaries}\label{sec:Convex-Analysis-Preliminaries}

For easy reference, we summarize key definitions from convex analysis following \citet{rockafellar1970convex}. Let $f:\R^n\to[-\infty,\infty]$ be an extended real-valued function on $\R^n$. The \emph{epigraph} of $f$ is
\[
\epi f:=\{(\bx,\mu)\in\R^n\times\R\mid f(\bx)\leqslant\mu\},
\]
the set of points on or above the graph of $f$. If $\epi f$ is convex in $\R^{n+1}$, then $f$ is \emph{convex}.
The \emph{effective domain} of $f$ is the set
\[
\dom\,f:=\left\{\bx\in\R^n\middle| f(\bx)<\infty\right\}=\left\{\bx\in\R^n\middle| (\bx,\mu)\in\epi f\text{ for some }\mu\right\}.
\]
If $f(\bx)>-\infty$ for all $\bx$ and $\dom\,f\neq\emptyset$, then $f$ is \emph{proper}.

The \emph{lower semi-continuous hull} of a function $f$ is the greatest lower semi-continuous minorant of $f$. The \emph{closure} of a convex function $f$, denoted $\cl\,f$, is defined as:
\[
\cl\,f :=
\begin{cases}
\text{lower semicontinuous hull of } f, & \text{if $f$ is nowhere } -\infty,\\
-\infty \text{ everywhere}, & \text{otherwise},
\end{cases}
\]
and $f$ is said to be \emph{closed} if it agrees with its closure, $\cl\,f=f$. 

For proper functions, convexity of the epigraph is equivalent to
\[
f\left(\left(1-\lambda\right)\bx+\lambda\by\right)\leqslant\left(1-\lambda\right)f\left(\bx\right)+\lambda f\left(\by\right)\text{ for all } \left(\bx,\by\right)\in\R^n\times\R^n,\quad\lambda\in\left[0,1\right].
\]
For proper convex functions, closedness is equivalent to lower semicontinuity.

For a convex function $f$, a vector $\by\in\R^n$ is a \emph{subgradient of $f$ at $\bx$} if
\[
f(\bz)\geqslant f(\bx)+\langle\by,\bz-\bx\rangle\quad\text{for all }\bz\in\R^n.
\]
The set of all such subgradients is the \emph{subdifferential} of $f$ at $\bx$:
\[
\partial f(\bx):=\left\{\by\in\R^n\middle|f(\bz)\geqslant f(\bx)+\langle\by,\bz-\bx\rangle\text{ for all }\bz\in\R^n\right\}.
\]
If $f$ is differentiable at $\bx$, then $\partial f(\bx)=\{\nabla f(\bx)\}$. Conversely, if $\partial f(\bx)$ is a singleton, then $f$ is differentiable at $\bx$, with gradient given by that unique subgradient.

The domain of the subdifferential is
\[
\dom\,\partial f:=\left\{\bx\in\R^n\middle|\partial f(\bx)\neq\emptyset\right\}.
\]
For proper convex $f$, \citep[Theorem 23.4]{rockafellar1970convex}:
\[
\interior(\dom\,f)\subseteq\dom\,\partial f\subseteq\dom\,f.
\]

The \emph{conjugate}\footnote{The conjugate $f^\ast$ is also known as the \emph{Fenchel transform} of $f$. The Fenchel transform generalizes the Legendre transform, which applies to smooth convex functions. See \citet[Chapter 26]{rockafellar1970convex} for details.} of $f$ is the function $f^\ast:\R^n\to[-\infty,\infty]$ defined by
\[
f^\ast(\by):=\sup_{\bx\in\R^n}\left\{\langle \bx,\by\rangle - f(\bx)\right\}.
\]
The conjugate $f^\ast$ of a convex function $f$ is a closed and convex function, proper if and only if $f$ is proper.  If $f$ is a closed proper convex function, then $\partial f^\ast$ is the correspondence inverse of $\partial f$, i.e., $\bx\in\partial f^\ast(\bx^\ast)$ if and only if $\bx^\ast\in\partial f(\bx)$.

\section{Proofs}\label{sec:Proofs}

The proofs of our main results rely on two lemmas concerning the link perturbation functions $\left\{ h_{e}\right\} _{e\in\cE}$ and the resulting perturbation function $H$, in turn, the proofs of which can be found at the end of this section. 
\begin{lem}
[\textbf{Link Perturbation Function Properties}]\label{lem:Perturbed-Cost-Contributions-Properties}
Each function $h_{e},e\in\cE,$ is closed proper convex, strictly
convex on $\R_{+}$, and satisfies $\lim_{\lambda\to\infty}h_{e}\left(\lambda z\right)/\lambda=\infty$
for each $z\neq0$.
\end{lem}
\begin{lem}
[\textbf{Perturbation Function Properties}]\label{lem:Perturbed-Cost-Properties}
The function $H$ is closed proper convex, has $\mathrm{dom}\,H=\R_{+}^{\cE}$,
is strictly convex on $\R_{+}^{\cE}$, and satisfies
$\lim_{\lambda\to\infty}H(\lambda\boldsymbol{z})/\lambda=\infty$
for $\boldsymbol{z}\neq\boldsymbol{0}_{\cE}$.
\end{lem}
With Lemmas \ref{lem:Perturbed-Cost-Contributions-Properties} and
\ref{lem:Perturbed-Cost-Properties} at hand, we are now ready to
prove Theorem \ref{thm:Primal-Properties}.
\begin{proof}
[\upshape\textsc{Proof of Theorem \ref{thm:Primal-Properties}}] Let $L$
denote the set of flows satisfying flow conservation
\begin{equation}\label{eq:DefinitionLSetAxEqb}
L:=\{ \bx\in\R^{\cE}|\bA\bx=\bb\}=\{\bx\in\R^{\cE}|\overline\bA\bx=\overline\bb\},    
\end{equation}
which is both closed and convex. Since the network is strongly connected,
there is a positive path $P:o\to d$ from the traveler's origin $o$ to the destination $d$. Fix such a positive path, and define the vector $\bx\in\R^{\cE}$ by setting $x_{e}$ equal to the number of times the link $e\in\cE$ is used when following this path. Then both $\bx\in L$
and $\bx\in\R_{+}^{\cE}=\dom\,H$ (Lemma \ref{lem:Perturbed-Cost-Properties}), implying that the traveler's problem (\ref{eq:Primal}) is feasible. The problem (\ref{eq:Primal}) is therefore equivalent to:
\[
\Minimize\,f\left(\bx\right)\text{ over flows }\bx\in\R^{\cE},
\]
where the objective function $f$ is defined on $\R^{\cE}$ by
\[
f\left(\bx\right):=\left\langle \bc,\bx\right\rangle +H\left(\bx\right)+\delta_{L}\left(\bx\right),\quad\bx\in\R^{\cE},
\]
and $\delta_{L}:\R^{\cE}\to\left[0,\infty\right]$
denotes the (convex) indicator of $L$, defined by
\begin{equation}\label{eq:ConvexIndicatorOfL}
\delta_{L}\left(\bx\right)=\begin{cases}
0, & \bx\in L,\\
\infty, & \bx\notin L.
\end{cases}    
\end{equation}
The criterion function $f$ is the sum of the three closed proper convex functions $\left\langle \bc,\cdot\right\rangle,H$ and $\delta_{L}$, whose effective domains $(\R^{\cE},\R_{+}^{\cE}$ and $L$, respectively) overlap. Hence, $f$ is closed proper convex \citep[Theorem 9.3]{rockafellar1970convex} with recession function $r_{f}$ of $f$ given by the sum $r_{\left\langle \bc,\cdot\right\rangle }+r_{H}+r_{\delta_{L}}$ of summand recession functions \citep[Theorem 9.3]{rockafellar1970convex}, each summand recession function being proper convex \citep[Theorem 8.5]{rockafellar1970convex}. As $\bzero_{\cE}\in\dom\,H$, the recession function $r_{H}$ can be evaluated as $r_{H}(\bz)=\lim_{\lambda\to\infty}H(\lambda\bz)/\lambda$ \citep[Theorem 8.5.2]{rockafellar1970convex}, so that $r_{H}(\bz)=\infty$ for each $\bz\neq\bzero_{\cE}$ by Lemma \ref{lem:Perturbed-Cost-Properties}. It follows from properness that $r_{f}\left(\bz\right)=\infty$ for each $\bz\neq\bzero_{\cE}$, meaning that $f$ is closed proper convex with no (nonzero) direction of recession. \citet[Theorem 27.2]{rockafellar1970convex} therefore shows that the optimal primal value $p=\inf f$ is finite and attained, and the primal optimal solution $\widehat{X}=\argmin f$ is non-empty and convex (and compact). Solution uniqueness then follows from strict convexity of $H$ on $\R_{+}^{\cE}$ (Lemma \ref{lem:Perturbed-Cost-Properties}).
\end{proof}
\begin{proof}
[\upshape\textsc{Proof of Lemma \ref{lem:Conjugate-Perturbed-Cost-Contributions-Properties}}]
Fix $e\in\cE$ and suppress the $e$ subscript to ease notation
$(h:=h_{e})$. Lemma \ref{lem:Perturbed-Cost-Contributions-Properties}
shows that $h$ is closed proper convex, which implies that its conjugate
$h^{\ast}$ is closed proper convex \citep[Theorem 12.2]{rockafellar1970convex}.
The subdifferentials of closed proper convex conjugate pairs are inverse
to each other, in that $\eta\in\partial h\left(\xi\right)$ if and
only if $\xi\in\partial h^{\ast}\left(\eta\right)$ \citep[Corollary 23.5.1]{rockafellar1970convex}.
Hence, $\mathrm{dom}\,\partial h=\mathrm{range}\,\partial h^{\ast}$
and $\mathrm{range}\,\partial h=\mathrm{dom}\,\partial h^{\ast}$.

Lemma \ref{lem:Perturbed-Cost-Contributions-Properties} also shows
that $h$ is strictly convex on $\mathrm{dom}\,h=\R_{+}$
and $\lim_{\lambda\to\infty}h(\lambda z)/\lambda=\infty$ for each
$z\neq0$. It follows that $h$ is both \emph{co-finite} \citep[p.~116]{rockafellar1970convex}
and \emph{essentially strictly convex} \citep[p.~253]{rockafellar1970convex},
which translate into $h^{\ast}$ being \emph{finite} (i.e., real-valued)
\citep[Corollary 13.3.1]{rockafellar1970convex} and \emph{essentially
differentiable} \citep[Theorem 26.3]{rockafellar1970convex}, respectively.
Finiteness and essential differentiability yields (ordinary) differentiability
\citep[p.~251]{rockafellar1970convex}. We identify the derivative
mapping $\nabla h^{\ast}:\R\to\R$ with the single-
and non-empty-valued correspondence $\partial h^{\ast}:\R\rightrightarrows\R$,
which is the (correspondence) inverse of $\partial h$, so that $\{\nabla h^{\ast}(\eta)\}=(\partial h)^{-1}(\eta)$
for each $\eta\in\R$. As $\mathrm{dom}\,h=\R_{+}$
and $\mathrm{int}(\mathrm{dom}\,h)\subseteq\mathrm{dom}\,\partial h\subseteq\mathrm{dom}\,h$
\citep[Theorem 23.4]{rockafellar1970convex}, we have $\R_{++}\subseteq\mathrm{range}\,\partial h^{\ast}=\mathrm{dom}\,\partial h\subseteq\R_{+}$.

It remains to show the expression (\ref{eq:ConjugateExplicitFormula})
for the conjugate $h^{\ast}$. Per closed proper convexity of $h$
and \citet[Theorem 23.5(a*)]{rockafellar1970convex}, the right-hand
side supremum in (\ref{eq:DefinitionfConjugate}) with $f=h$ is attained at $\xi$
if and only if $\xi\in\partial h^{\ast}\left(\eta\right)$. As $h^{\ast}$
is differentiable, the maximizer exists and is uniquely given by $\xi=\nabla h^{\ast}\left(\eta\right)$,
the single element in $(\partial h)^{-1}\left(\eta\right)$. Treating
the latter singleton as the element it contains, we arrive at the
formula in (\ref{eq:ConjugateExplicitFormula}).
\end{proof}

\begin{proof}[\upshape\textsc{Proof of Lemma \ref{lem:Conjugate-Perturbed-Cost-Properties}}]
Since $H(\bx)=\sum_e h_e(x_e)$ is additively separable, the expression in \eqref{eq:Conjugate-Perturbed-Cost-Explicit} follows from definition \eqref{eq:DefinitionfConjugate} of a conjugate. Convexity and differentiability of $H^*$ follow from Lemma \ref{lem:Conjugate-Perturbed-Cost-Contributions-Properties}. The remaining claims are immediate from the additive separability $H^*(\by)=\sum_e h_e^*(y_e)$.
\end{proof}

\begin{proof}[\upshape\textsc{Proof of Theorem \ref{thm:PstarIsPDual}}]
To arrive at the form in (\ref{eq:Dual}), rewrite the dual function $g$
at $\bu\in\R^{\cV\setminus\{d\}}$ as follows:
\begin{align*}
g\left(\bu\right) &= \inf_{\bx\in\R^{\cE}}\left\{\langle\bc,\bx\rangle+H\left(\bx\right)+\langle\bu,\bb-\bA\bx\rangle\right\}= \langle\bb,\bu\rangle+\inf_{\bx\in\R^{\cE}}\left\{ \bc^{\top}\bx-\bu^{\top}\bA\bx+H\left(\bx\right)\right\} \\
 &= \langle\bb,\bu\rangle-\sup_{\bx\in\R^{\cE}}\left\{ \bx^{\top}(\bA^{\top}\bu-\bc)-H\left(\bx\right)\right\}=\langle\bb,\bu\rangle-H^*(\bA^{\top}\bu-\bc),
\end{align*}
where the last equality stems from the definition \eqref{eq:DefinitionfConjugate} of a conjugate. Using Lemma \ref{lem:Conjugate-Perturbed-Cost-Properties} and \citet[Theorem 5.7]{rockafellar1970convex}, we see that \eqref{eq:Dual} is (also implicitly) unconstrained and involves  maximization of a differentiable concave criterion function.

\end{proof}

\begin{proof}
[\upshape\textsc{Proof of Theorem \ref{thm:Duality}}] The proof of the theorem is twofold. We first verify a weak version of Slater's condition, which will serve as our constraint qualification, and then apply Karush-Kuhn-Tucker theory to arrive at the claimed properties.

Towards Slater's condition, recall the numbering $\cE=\{e_1,e_2,\dotsc,e_{|\cE|}\}$ of the links, and the notation $v(e)$ for the initial node of link $e\in\cE$, to define nodes $v^0:=o$, $v^j:=v(e_j),j\in\{1,2,\dotsc,|\cE|\}$, and $v^{|\cE|+1}:=o$. Given strong network connectedness, we can find $|\cE|+1$ positive paths $P_{j+1}:v^{j}\to v^{j+1}$, from $v^{j}$ to $v^{j+1},j\in\{0,1,\dotsc,|\cE|\}$.
Fixing such a finite collection of positive paths $P_1,P_2,\dotsc,P_{\left|\cE\right|+1}$, we piece together a positive circuit $P_{++}:o\to o$, which starts/ends at the traveler's origin $o$, passing every link along the way. Define a flow $\bx'\in\R^{\cE}$ by setting $x'_{e}$ equal to the number of times the link $e\in\cE$ is used when following the circuit $P_{++}$. Then both $\bA\bx'=\mathbf{0}_{\cV}$ and $\bx'\in\R^{\cE}_{++}$.

Similarly, fix a positive path $P:o\to d$ from the origin $o$ to the destination $d$, and define the flow $\bx''\in\R^{\cE}$ by setting $x''_{e}$ equal to the number of times the link $e\in\cE$ is used when following the path $P$. Then both $\bA\bx''=\bb$ and $\bx''\in\R^{\cE}_{+}$. It follows that the combined flow defined by $\bx''':=\bx'+\bx''$
is both primal feasible and satisfies $\bx'''\in\R_{++}^{\cE}$,
i.e., it lies in the interior of the effective domain of $\bx\mapsto\langle\bc,\bx\rangle+H(\bx)$.

As the optimal primal value $p$ in (\ref{eq:Primal}) is finite (Theorem
\ref{thm:Primal-Properties}), and (\ref{eq:Primal}) has only (explicit)
linear equality constraints, the existence of a feasible solution
$\bx'''\in\R_{++}^{\cE}$ implies the existence of a
\emph{Kuhn--Tucker} vector for (\ref{eq:Primal}) \citep[Corollary 28.2.2]{rockafellar1970convex}.
Any Kuhn--Tucker vector attains the supremum of the (Lagrangian) dual problem defined in \eqref{eq:DualProblemProofs}, which by Theorem \ref{thm:PstarIsPDual} takes the form in \eqref{eq:Dual}. It follows from \citet[Theorem 28.4]{rockafellar1970convex} that the optimal dual solution $(\widehat{U})$ is non-empty, and that the primal and dual optimal values coincide, $p^{\ast}=p\in\R$, where finiteness follows from Theorem \ref{thm:Primal-Properties}. That $\widehat{U}$ is a closed set follows from $H^*$ being finite convex (Lemma \ref{lem:Conjugate-Perturbed-Cost-Properties}), hence a continuous function \citep[Corollary 10.1.1]{rockafellar1970convex}. 

For the expression (\ref{eq:PrimalFromDual}), fix an optimal dual solution $\widehat{\bu}\in \widehat U$.
Again using \citet[Theorem 28.4]{rockafellar1970convex} and the fact that
$\widehat{\bu}$ is a Kuhn--Tucker vector for (\ref{eq:Primal}),
we see that the optimal flow satisfies
\[
\widehat{\bx}\in\argmin_{\bx\in\R^{\cE}}\cL\left(\widehat{\bu},\bx\right),
\]
with the Lagrangian $\cL$ given in \eqref{eq:LagrangianDefinition}. As
\begin{align*}
\inf_{\bx\in\R^{\cE}}\cL\left(\widehat{\bu},\bx\right) & =\inf_{\bx\in\R^{\cE}}\left\{ \langle\bc,\bx\rangle+H\left(\bx\right)+\langle\widehat{\bu},\bb-\bA\bx\rangle\right\} \\
 &= \langle\bb,\widehat{\bu}\rangle-\sup_{\bx\in\R^{\cE}}\left\{ \langle \bx, \bA^{\top}\widehat{\bu}-\bc \rangle -H\left(\bx\right)\right\}.
\end{align*}
\citet[Theorem 23.5(a*)]{rockafellar1970convex} and $H$ being closed proper convex (Lemma \ref{lem:Perturbed-Cost-Properties}) combine to show that latter supremum is achieved at $\widehat{\bx}$ if and only if $\widehat{\bx}\in\partial H^{\ast}(\bA^{\top}\widehat{\bu}-\bc)$. Since $H^{\ast}$ is differentiable (Lemma \ref{lem:Conjugate-Perturbed-Cost-Properties}), this inclusion means precisely \eqref{eq:PrimalFromDual}.
\end{proof}

\begin{proof}[\upshape\textsc{Proof of Theorem \protect\ref{thm:Hessian-Dual-Function}}]

By \eqref{eq:gradient}, the gradient of $g$ at $\bu$ is 
\(
\nabla g\left( \bu\right) =\bb-\bA \nabla H^*(\bA^{\top }\bu-\bc).
\)
By Theorem \ref{thm:PstarIsPDual}, $g$ is concave. Hence, by Alexandrov's
theorem \citep{Alexandroff1939}, it is almost everywhere twice differentiable. At a point $\bu$ of twice differentiability, the Hessian of $g$ is precisely the expression in \eqref{eq:Hessian-Dual-Function}.
\end{proof}

\begin{proof}[\upshape\textsc{Proof of Theorem \ref{thm:CostFlowAndValueProperties}}]
We establish Items \ref{enu:CostValueConcavity}, \ref{enu:CostFlowContinuityAndMonotonicity} and \ref{enu:CostSupportConstancy} in that order.

\smallskip
\noindent\emph{Item \ref{enu:CostValueConcavity}.}
Consider the convex indicator $\delta_L:\R^\cE\to[0,\infty]$ of the feasibility set $L$ in \eqref{eq:DefinitionLSetAxEqb}. Absorbing the conservation constraint into the objective via $\delta_L$ and using the definition of the conjugate, for any $\bc\in\R^\cE$ we can write
\begin{align*}
    p(\bc)
    &=\inf_{\bx\in\R^\cE}\left\{\langle \bc,\bx\rangle +H(\bx)+\delta_{L}(\bx)\right\}\\
    &=-\sup_{\bx\in\R^\cE}\left\{\langle -\bc,\bx\rangle -\left[H(\bx)+\delta_{L}(\bx)\right]\right\}
    =-\left(H+\delta_L\right)^*(-\bc).
\end{align*}
As argued in the proof of Theorem \ref{thm:Primal-Properties}, $H+\delta_L$ is closed proper convex. Its conjugate $(H+\delta_L)^*$ is therefore convex \citep[Theorem 12.2]{rockafellar1970convex}, and composition with the affine map $\bc\mapsto -\bc$ preserves convexity \citep[Theorem 5.7]{rockafellar1970convex}. Hence $\bc\mapsto (H+\delta_L)^*(-\bc)$ is convex, so $\bc\mapsto p(\bc)=- (H+\delta_L)^*(-\bc)$ is concave.
 
\smallskip
\noindent\emph{Item \ref{enu:CostFlowContinuityAndMonotonicity}.}
Fix $\bc\in\R^\cE$. Closed proper convexity of $H+\delta_L$ and \citet[Theorem 23.5(a*)]{rockafellar1970convex} yield that $\bz\mapsto \langle -\bc,\bz\rangle-[H(\bz)+\delta_L(\bz)]$ attains its supremum at $\bx$ if and only if $\bx\in \partial(H+\delta_L)^*(-\bc)$. By Theorem \ref{thm:Primal-Properties}, the primal problem has a unique minimizer $\hbx(\bc)$; equivalently, $\hbx(\bc)$ is the unique maximizer of the displayed concave maximization problem. Hence $\partial(H+\delta_L)^*(-\bc)$ is a singleton. By \citet[Theorem 25.1]{rockafellar1970convex}, $(H+\delta_L)^*$ is differentiable at $-\bc$ and
\(
\hbx(\bc)=\nabla (H+\delta_L)^*(-\bc).
\)
Since $\bc$ is arbitrary, $(H+\delta_L)^*$ is differentiable on all of $\R^\cE$, and its gradient $\nabla(H+\delta_L)^*$ is continuous \citep[Corollary 25.5.1]{rockafellar1970convex}. Therefore the cost--flow map
\(
\bc\mapsto \hbx(\bc)=\nabla(H+\delta_L)^*(-\bc)
\)
is continuous as the composition of a continuous map and the linear sign-flip $\bc\mapsto -\bc$.

To obtain anti-monotonicity, note that the gradient of any differentiable convex function is monotone \citep[Theorem 4.1.4]{hiriart2004fundamentals}. Hence, for all $\by,\by'\in\R^\cE$,
\[
\langle \nabla(H+\delta_L)^*(\by')-\nabla(H+\delta_L)^*(\by),\by-\by\rangle \geqslant 0.
\]
Applying this with $\by'=-\bc'$ and $\by=-\bc$ and recalling $\hbx(\bc)=\nabla(H+\delta_L)^*(-\bc)$ gives the claimed anti-monotonicity.

\smallskip
\noindent\emph{Item \ref{enu:CostSupportConstancy}.}
Define $\balpha:\R^\cE\to\{0,1\}^\cE$ coordinatewise by
\[
\alpha_e(\bc):=\lim_{n\to\infty}\hx_e(\bc)^{1/n}=\bone\{\hx_e(\bc)>0\},\quad e\in\cE.
\]
Then $\alpha_e(\bc)=1$ if and only if $e\in\cS(\bc)$, so local constancy of $\cS(\cdot)$ is equivalent to local constancy of $\balpha(\cdot)$.

Since $\hbx(\cdot)$ is continuous [Item \ref{enu:CostFlowContinuityAndMonotonicity}] and $\xi\mapsto \xi^{1/n}$ is continuous on $\R_+$, each map $\bc\mapsto \hx_e(\bc)^{1/n}$ is continuous. Thus each coordinate function $\alpha_e$ is the pointwise limit of continuous functions and is therefore of Baire class $1$. Because $\cE$ is finite, $\balpha=(\alpha_e)_{e\in\cE}$ is also of Baire class $1$ as a map into the finite product space $\{0,1\}^\cE$.

Consider the set defined as
\[
\cC_0:=\{\bc\in\R^\cE|\balpha \text{ is continuous at }\bc\}.
\]
By the classical theorem that Baire class $1$ functions have a meager set of discontinuity points \citep[Theorem 24.14]{kechris_classical_1995}, the complement $\R^\cE\setminus \cC_0$ is meager in $\R^\cE$. In particular, $\cC_0$ is dense.

We next show that $\balpha$ is locally constant at every $\bc\in\cC_0$. Since $\cE$ is finite, the codomain $\{0,1\}^\cE$ is a finite discrete topological space; in particular, the singleton $\{\balpha(\bc)\}$ is open. Continuity of $\balpha$ at $\bc$ therefore implies that there exists an open neighborhood $\cN$ of $\bc$ such that
\[
\balpha(\bc')\in \{\balpha(\bc)\}\quad\text{for all }\bc'\in\cN,
\]
i.e., $\balpha$ is constant on $\cN$. Equivalently, $\cS(\bc')=\cS(\bc)$ for all $\bc'\in\cN$, proving local constancy of the cost--support function on $\cC_0$.

Finally, the preceding argument also shows that $\cC_0$ is open: if $\bc\in\cC_0$, then $\balpha$ is constant on some neighborhood of $\bc$, hence continuous at every point of that neighborhood. Thus $\cC_0$ is open and dense in $\R^\cE$, completing the proof.
\end{proof}

\begin{proof}[\upshape\textsc{Proof of Theorem \ref{thm:CostFlowLipschitzness}}]
The proof will rely on the notion of strong convexity in Banach spaces as defined in \citet[Section 3.5]{zalinescu2002convex}. The Banach-space formulation is convenient for handling weighted $\ell_2$ norms and their duals in a unified way. Specifically, let $\left\|\cdot\right\|_\mu:\R^\cE\to\R_+$ denote the $\{\mu_e\}$-weighted $\ell_2$ norm $\left\|\bx\right\|_{\mu}:=(\sum_{e\in\cE}\mu_e|x_e|^2)^{1/2}$. As $\left\|\cdot\right\|_\mu$ is equivalent to the ordinary Euclidean norm $(\left\|\cdot\right\|)$, the space $X:=(\R^\cE,\left\|\cdot\right\|_\mu)$ is a Banach space. The dual to this space is $X^*:=(\R^\cE,\left\|\cdot\right\|_{\mu^{-1}})$, where $\left\|\cdot\right\|_{\mu^{-1}}:\R^\cE\to\R_+$ denotes the norm dual to $\left\|\cdot\right\|_\mu$, which is the reciprocally weighted $\ell_2$ norm given by $\left\|\by\right\|_{\mu^{-1}}:=(\sum_{e\in\cE}\mu_e^{-1}|y_e|^2)^{1/2}$. Since each $h_e$ is $\mu_e$-strongly convex with $\dom\,h_e=\R_+$, for any $(\bx,\bx')\in\dom\,H\times\dom\,H=\R^\cE_+\times\R^\cE_+$ and any $\lambda\in[0,1]$,
\begin{align*}
    H\left(\left(1-\lambda\right)\bx+\lambda\bx'\right)
    &=\sum_{e\in\cE}h_e\left(\left(1-\lambda\right)x_e+\lambda x_e'\right)\\
    &\leqslant \sum_{e\in\cE}\left[\left(1-\lambda\right)h_e\left(x_e\right)+\lambda h_e\left( x_e'\right) - \frac{\mu_e}{2}\lambda\left(1-\lambda\right)\left|x_e-x_e'\right|^2\right]\\
    &= \left(1-\lambda\right)\sum_{e\in\cE}h_e\left(x_e\right)+\lambda \sum_{e\in\cE}h_e\left( x_e'\right) - \frac{1}{2}\lambda\left(1-\lambda\right)\sum_{e\in\cE} \mu_e\left|x_e-x_e'\right|^2\\
    &=\left(1-\lambda\right)H\left(\bx\right)+\lambda H\left(\bx'\right) - \frac{1}{2}\lambda\left(1-\lambda\right)\left\Vert\bx-\bx'\right\Vert^2_{\mu}.
\end{align*}
Viewing $H$ as a closed proper convex function defined on the Banach space $X$, the inequality shows that $H$ is $1$-strongly convex (with respect to $\left\|\cdot\right\|_\mu$). As argued in the proof of Theorem \ref{thm:Primal-Properties}, the intersection of $\dom\,H=\R^\cE_+$ and the (affine) set $L$ in \eqref{eq:DefinitionLSetAxEqb} is non-empty, implying that $H+\delta_L$ is closed proper convex. As convex combinations of elements in $\dom(H+\delta_L)=\R_+^\cE\cap L$ stay in $L$, addition of $\delta_L$ has no impact on the strong convexity inequality on the effective domain. The previous display therefore implies that $H+\delta_L$ is $1$-strongly convex (with respect to $\left\|\cdot\right\|_\mu$). It therefore follows from \citet[Corollary 3.5.11(i),(x) and Remark 3.5.3]{zalinescu2002convex}, that, viewed as a function defined on $X^*$,  $(H+\delta_L)^*$ is differentiable and satisfies
\[
\left\|\nabla\left(H+\delta_L\right)^*\left(\by\right)-\nabla\left(H+\delta_L\right)^*\left(\by'\right)\right\|_{\mu}\leqslant\left\|\by-\by'\right\|_{\mu^{-1}}\text{ for }\left(\by,\by'\right)\in\R^\cE\times\R^\cE.
\]
The optimal flow $\hbx\left(\cdot\right)$ can be expressed as $\hbx(\bc)=\nabla(H+\delta_L)^*(-\bc)$, cf.~the proof of Theorem \ref{thm:CostFlowAndValueProperties}. The cost--flow mapping $\bc\mapsto\hbx(\bc)$ therefore satisfies,
\[
\left\|\hbx\left(\bc\right)-\hbx\left(\bc'\right)\right\|_{\mu}\leqslant\left\|\bc-\bc'\right\|_{\mu^{-1}}\text{ for }\left(\bc,\bc'\right)\in\R^\cE\times\R^\cE.
\]
Unpacking the norms, we arrive at the inequality \eqref{eq:CostFlowLipschitzWeighted}. The inequality \eqref{eq:CostFlowLipschitzian} follows from $\mu_{\min}=\min_{e\in\cE}\mu_e$ and the norm relations
$\left\|\bx\right\|\leqslant \mu_{\min}^{-1/2}\left\|\bx\right\|_{\mu}$ and $\left\|\by\right\|_{\mu^{-1}}\leqslant\mu_{\min}^{-1/2}\left\|\by\right\|$.
\end{proof}

\begin{proof}[\upshape\textsc{Proof of Theorem \ref{thm:CostFlowDerivative}}]
Fix $\obc$ satisfying the local regularity conditions. By Assumption \ref{enu:SupportStability} and continuity of the cost--flow mapping $\bc\mapsto\hbx(\bc)$ (cf.~Theorem \ref{thm:CostFlowAndValueProperties}), shrinking the neighborhood $\cN$ of $\obc$ if needed ensures both $\hbx_{\cI}(\bc)=\mathbf 0$ and $\hbx_{\cA}(\bc)\in\R_{++}^{\cA}$ for all $\bc\in\cN$. For any $\bc\in\cN$, the traveler's original problem is therefore equivalent to the reduced equality-constrained problem on the active links
\begin{equation}\label{eq:ReducedProblem}
\left.\begin{aligned}
\Minimize\quad 
& \left\langle \bc_{\cA},\bx_{\cA}\right\rangle + H_{\cA}\left(\bx_{\cA}\right)
\text{ over }\bx_{\cA}\in\R^{\cA}\\
\text{satisfying}\quad 
& \bA_{\cA}\bx_{\cA}=\bb,
\end{aligned}
\right\}
\end{equation}
with $\bx_{\cI}\equiv \mathbf 0$. Consequently, on $\cN$ the cost--flow mapping $\bc\mapsto\hbx(\bc)$ has the form $\hbx(\bc)=(\hbx_{\cA}(\bc),\mathbf 0_{\cI})$ for $\bc=(\bc_{\cA},\bc_{\cI})$ with $\hbx_{\cA}\left(\cdot\right)$ depending on $\bc$ only through $\bc_{\cA}$ and $\hbx_{\cI}\left(\cdot\right)$ identically zero. Hence, if it exists, the full Jacobian must have the block form \eqref{eq:FullJacobianBlock} with $\nabla_{\bc_{\cA}}\hbx_{\cA}(\obc)$ as the only potentially nonzero block. Even if the full Jacobian does not exist, we have assured that the zero blocks are as stated.  It remains to show the existence of $\nabla_{\bc_{\cA}}\hbx_{\cA}(\obc)$.

To this end, consider $k=\dim\,\ker(\bA_\cA)$. We split the argument into two cases depending on whether $k=0$ or $k\geqslant 1$. In case $k=0$, the nullspace $\ker(\bA_{\cA})$ is trivial. It follows that $\bA_\cA$ has full column rank, implying that the solution to \eqref{eq:ReducedProblem} is uniquely given by the constraint. Since $\hbx_\cA(\obc)$ is feasible, it must be the unique solution for any $\bc\in\cN$. Hence, $\bc\mapsto\hbx_\cA(\bc)$ is locally constant, and therefore differentiable with $\nabla_{\bc_\cA}\hbx_\cA(\obc)=\bzero_{\cA\times\cA}$.

In the case where $k\geqslant1$, the nullspace $\ker(\bA_\cA)$ is nontrivial. Fix any $\obx_{\cA}\in\R_{++}^{\cA}$ satisfying $\bA_{\cA}\obx_{\cA}=\bb$ ($\hbx_{\cA}(\obc)$ will do), and let $\bN_{\cA}\in\R^{\cA\times k}$ be any matrix with columns forming a basis of $\ker(\bA_{\cA})$, so that $\bA_{\cA}\bN_{\cA}=\bzero$ and $k=\rank(\bN_{\cA})=\dim\ker(\bA_{\cA})\geqslant1$. Since $\bA_{\cA}\obx_{\cA}=\bb$ and $\bA_{\cA}\bN_{\cA}=\mathbf 0$, the $\bx_{\cA}$ feasible in \eqref{eq:ReducedProblem} can be expressed as
\begin{equation}\label{eq:NullspaceReparameterization}
\bx_{\cA}=\obx_{\cA}+\bN_{\cA}\bz,\qquad \bz\in\R^{k}.
\end{equation}
Using this reparameterization, define a reduced objective function $\phi:\R^{k}\times\R^{\cA}\to(-\infty,\infty]$ by
\[
\phi(\bz,\bc_{\cA}):=\langle \bc_{\cA},\obx_{\cA}+\bN_{\cA}\bz\rangle
+H_{\cA}(\obx_{\cA}+\bN_{\cA}\bz).
\]
Because $H_{\cA}$ is strictly convex on $\R_{+}^{\cA}$ and $\bN_{\cA}$ has full column rank, for any fixed $\bc_{\cA}$, the mapping $\bz\mapsto\phi(\bz,\bc_{\cA})$ is strictly convex on its effective domain. Let $\hbz(\obc_{\cA})\in\R^k$ be such that $\hbx_{\cA}(\obc)=\obx_{\cA}+\bN_{\cA}\hbz(\obc_{\cA})$. It then follows from the reparameterization that $\hbz(\obc_{\cA})$ is the unique minimizer of $\phi(\cdot,\obc_{\cA})$. 

Since $\hbx_{\cA}(\obc)=\obx_{\cA}+\bN_{\cA}\hbz(\obc_{\cA})\in\R_{++}^{\cA}=\mathrm{int}(\dom\,H_{\cA})$, and Assumption \ref{enu:LocalCurvature} implies that $H_{\cA}$ is twice continuously differentiable in the open neighborhood $\times_{e\in \cA}\cN_e\subset\R^{\cA}_{++}$ of $\hbx_{\cA}(\obc)$, $\phi(\cdot,\obc_{\cA})$ is differentiable at $\hbz(\obc_{\cA})$. Hence, $\hbz(\obc_{\cA})$ satisfies the first-order condition for a minimum in gradient form, $\nabla_{\bz}\phi(\hbz(\obc_{\cA}),\obc_{\cA})=\bzero$.

Since the mapping $\bz\mapsto \obx_{\cA}+\bN_{\cA}\bz$ is continuous, the set $U:=\{\bz\in\R^k|\obx_{\cA}+\bN_{\cA}\bz\in\times_{e\in \cA}\cN_e\}$ is open and contains $\hbz(\obc_{\cA})$. It follows that the function $\bG:U\times\R^{\cA}\to\R^k$ given in
\begin{equation}\label{eq:FOCz}
\bG(\bz,\bc_{\cA}):=\nabla_{\bz}\phi(\bz,\bc_{\cA})
=\bN_{\cA}^\top\left(\bc_{\cA} + \nabla H_{\cA}\left(\obx_{\cA}+\bN_{\cA}\bz\right)\right)
\end{equation}
is well defined. Then $\bG(\hbz(\obc_{\cA}),\obc_{\cA})=\bzero$, and $\bG$ is continuously differentiable on its domain with
\[
\nabla_{\bz}\bG(\hbz(\obc_{\cA}),\obc_{\cA})
=
\bN_{\cA}^\top \nabla^2 H_{\cA}(\hbx_{\cA}(\obc))\,\bN_{\cA}
=
\bN_{\cA}^\top \bD_{\cA}\,\bN_{\cA}.
\]

Assumption \ref{enu:LocalCurvature} further implies that $\bD_{\cA}(\obc)$ is positive definite, which by $\bN_{\cA}$ having full column rank implies that
$\bN_{\cA}^\top \bD_{\cA}\,\bN_{\cA}$ is positive definite and, thus, invertible. The Implicit Function Theorem (IFT)
therefore yields open neighborhoods
\[
U_0\subset U \ \text{of }\hbz(\obc_{\cA})
\ \text{and}\ 
V_0\subset\R^{\cA}\ \text{of }\obc_{\cA},
\]
as well as a unique continuously differentiable function $\bz^\star:V_0\to U_0$ such that
\[
\bG(\bz^\star(\bc_{\cA}),\bc_{\cA})=\mathbf 0\;\text{for all }\bc_{\cA}\in V_0\;\text{and}\;\bz^\star(\obc_{\cA})=\hbz(\obc_{\cA}).
\]
Define a mapping $\bx_{\cA}^\star(\cdot)$ on $V_0$ by
\[
\bx^\star_{\cA}(\bc_{\cA}):=\obx_{\cA}+\bN_{\cA}\bz^\star(\bc_{\cA}).
\]
Then $\bx^\star_{\cA}(\bc_{\cA})\in\R_{++}^{\cA}$ for all $\bc_{\cA}\in V_0$,
and $\bc_{\cA}\mapsto\bx^\star_{\cA}(\bc_{\cA})$ is continuously differentiable.
For each $\bc_{\cA}\in V_0$, we have $\bz^\star(\bc_{\cA})\in U$, hence $\obx_{\cA}+\bN_{\cA}\bz^\star(\bc_{\cA})\in\times_{e\in \cA}\cN_e\subset\R^{\cA}_{++}$, implying that $\phi(\cdot,\bc_{\cA})$ is differentiable at $\bz^\star(\bc_{\cA})$. The first-order condition \eqref{eq:FOCz} therefore
characterizes the unique minimizer of \eqref{eq:ReducedProblem} via the reparameterization \eqref{eq:NullspaceReparameterization}, so $\bx^\star_{\cA}(\bc_{\cA})$ solves the reduced problem for each $\bc_{\cA}\in V_0$. Further restricting
to the open neighborhood
\(
\cN_0:=\cN\cap\{\bc\in\R^\cE| \bc_{\cA}\in V_0\}
\)
of $\obc$, $\bx_{\cA}^\star\left(\cdot\right)$ therefore locally agrees with the cost--flow mapping, in the sense that $\hbx(\bc)=(\bx_{\cA}^\star(\bc_{\cA}),\bzero_{\cI})$ for all $\bc\in\cN_0$.  Hence, $\bc\mapsto\hbx(\bc)$ is continuously differentiable on $\cN_0$.

To compute the Jacobian on the active set, observe that the Jacobian of $\bG$ at $(\hbz(\obc_{\cA}),\obc_{\cA})$ is
\[
\nabla \bG(\hbz(\obc_{\cA}),\obc_{\cA})=[\nabla_{\bz} \bG(\hbz(\obc_{\cA}),\obc_{\cA}),\nabla_{\bc_{\cA}} \bG(\hbz(\obc_{\cA}),\obc_{\cA})]=\left[\bN_{\cA}^\top\bD_{\cA}\bN_{\cA},\bN_{\cA}^\top\right],
\]
so that the IFT further yields
\[
\nabla \bz^\star(\obc_{\cA})=-\left[\nabla_{\bz} \bG(\hbz(\obc_{\cA}),\obc_{\cA})\right]^{-1}\nabla_{\bc_{\cA}} \bG(\hbz(\obc_{\cA}),\obc_{\cA})=-\left(\bN_{\cA}^\top\bD_{\cA}\bN_{\cA}\right)^{-1}\bN_{\cA}^\top.
\]
Using $\hbx_{\cA}(\bc)=\bx_{\cA}^\star(\bc_{\cA})=\obx_{\cA}+\bN_{\cA}\bz^\star(\bc_{\cA})$ on $\cN_0$, we get 
\[
\nabla_{\bc_{\cA}}\hbx_{\cA}(\obc)=\bN_{\cA}\nabla \bz^\star(\obc_{\cA})=-\bN_{\cA}\left(\bN_{\cA}^\top \bD_{\cA}\,\bN_{\cA}\right)^{-1}\bN_{\cA}^\top,
\]
which is the expression in \eqref{eq:JacobianBlockNullspaceForm}. The equivalence between the expressions \eqref{eq:JacobianBlockNullspaceForm} and \eqref{eq:JacobianBlockLaplacianForm} follow from an application of Lemma \ref{lem:ProjectorMP} stated below.
\end{proof}

\begin{lem}\label{lem:ProjectorMP}
Let $\bA\in\R^{m\times n}$ be arbitrary and let
$\bD\in\R^{n\times n}$ be symmetric positive definite.
Let $k:=\dim\ker(\bA)$. If $k=0$, define $\bP:=\bzero_{n\times n}$. If $k\geqslant1$, let $\bN\in\R^{n\times k}$ have columns forming a basis of $\ker(\bA)$ and define
\[
\bP:=\bN(\bN^\top\bD\bN)^{-1}\bN^\top.
\]
Then $\bP$ admits the representation
\begin{equation}\label{eq:ProjectorMPSchur}
\bP
=
\bD^{-1}
-
\bD^{-1}\bA^\top(\bA\bD^{-1}\bA^\top)^{\dagger}\bA\bD^{-1},
\end{equation}
where $(\cdot)^{\dagger}$ denotes Moore--Penrose pseudoinversion.
\end{lem}

\begin{proof}[\upshape\textsc{Proof of Lemma \ref{lem:ProjectorMP}}]
Consider first the case of a \emph{trivial nullspace}, $\ker(\bA)=\{\bzero_n\}$. Then $k=0$, and $\bA$ has full column rank $n$. Consider the symmetric square root $\bS:=\bD^{-1/2}$ of $\bD^{-1}$, and set $\bB:=\bA\bS$. Since $\bA$ has full column rank, and $\bS$ is invertible, $\bB$ has full column rank $n$. Defining $\bL:=\bA\bD^{-1}\bA^\top$, we therefore have
\(
\bL=\bA\bS\bS\bA^\top
=\bB\bB^\top.
\)
Since $\bB$ has full column rank, conduct a singular value decomposition of $\bB$ to get $\bB=\bU\bSigma\bV^\top$ with $\bSigma\in\R^{m\times n}$ having positive diagonal entries. Then $(\bB\bB^\top)^\dagger=\bU(\bSigma\bSigma^\top)^\dagger\bU^\top$, and hence $\bB^\top(\bB\bB^\top)^\dagger \bB=\bV\bSigma^\top(\bSigma\bSigma^\top)^\dagger\bSigma\bV^\top=\bI_n$. Hence,
\[
\begin{aligned}
\bD^{-1}\bA^\top \bL^\dagger \bA \bD^{-1}
&=
\bD^{-1}\bA^\top(\bB\bB^\top)^\dagger \bA \bD^{-1} 
=
\bS\,\bB^\top(\bB\bB^\top)^\dagger \bB\,\bS 
=
\bS \bI_n \bS
=
\bD^{-1}.
\end{aligned}
\]
Therefore the Moore--Penrose representation yields
\[
\bD^{-1}
-
\bD^{-1}\bA^\top(\bA\bD^{-1}\bA^\top)^\dagger\bA\bD^{-1}
=
\bD^{-1}-\bD^{-1}
=
\mathbf 0_{n\times n}=\bP,
\]
as claimed.

Consider next the case of a \emph{nontrivial nullspace}, so that $k\geqslant1$, and $\bA$ has column rank less than $n$. Fix $\bv\in\R^n$ and consider the strictly convex quadratic program
\begin{equation}\label{eq:QPProjectorMP}
\Minimize\;\frac12\,\bx^\top\bD\bx-\bv^\top\bx\;\text{over}\;\bx\in\R^n\;\text{satisfying}\;\bA\bx=\bzero_m.
\end{equation}
Since $\bD$ is symmetric positive definite, and the constraint set is nonempty, the problem \eqref{eq:QPProjectorMP}
has a unique minimizer, which we denote by $\bx^\star(\bv)$. The $\bx$ feasible in \eqref{eq:QPProjectorMP} can be expressed as $\bx=\bN\bz,\bz\in\R^k$. Using this reparameterization, \eqref{eq:QPProjectorMP} is equivalent to the problem
\[
\Minimize\;\frac12\,\bz^\top(\bN^\top\bD\bN)\bz-\bz^\top\bN^\top\bv\;\text{over}\;\bz\in\R^k.
\]
Because $\bD$ is symmetric positive definite, and $\bN$ has full column rank, $\bN^\top\bD\bN$ is symmetric
positive definite and hence invertible. The (unique) minimizer in the previous display satisfies
$(\bN^\top\bD\bN)\bz=\bN^\top\bv$, which by the reparameterization implies that
\[
\bx^\star(\bv)=\bN(\bN^\top\bD\bN)^{-1}\bN^\top\bv=\bP\bv.
\]

For the pseudoinverse representation of $\bP$, we show that the linear map in \eqref{eq:ProjectorMPSchur} produces the same minimizer, thereby identifying the two representations of $\bP$. To this end, recall the matrix $\bL=\bA\bD^{-1}\bA^\top$, which is symmetric positive semidefinite, but not necessarily invertible. Let
\[
\bx(\bv):=\left[\bD^{-1}-\bD^{-1}\bA^\top\bL^{\dagger}\bA\bD^{-1}\right]\bv.
\]
We argue via the KKT conditions that $\bx(\bv)$ is actually the solution $\bx^\star(\bv)$ to
\eqref{eq:QPProjectorMP}.

Arguing as in the case of the trivial nullspace, for $\bS=\bD^{-1/2}$ and $\bB=\bA\bS$, we have $\bL=\bB\bB^\top$. A fact from linear algebra is that for any matrix $\bB$, one has $\range(\bB)=\range(\bB\bB^\top)$.
Consequently, since right multiplication by the invertible matrix $\bS$ does not change the column space, for our $\bB$ we get
\[
\range(\bL)=\range(\bB\bB^\top)=\range(\bB)=\range(\bA\bS)
=\range(\bA).
\]
The range equality implies that $\bA\bD^{-1}\bv \in \range(\bL)$. Since $\bL\bL^{\dagger}$ is the orthogonal projector onto its range, we conclude that $\bL\bL^{\dagger}\bA\bD^{-1}\bv=\bA\bD^{-1}\bv$. Consequently,
\[
\bA\bx(\bv)
=
\bA\bD^{-1}\bv
-
\bL\bL^{\dagger}\bA\bD^{-1}\bv
=
\bzero_m,
\]
which establishes the feasibility of $\bx(\bv)$.

Second, define candidate Lagrange multipliers by $\blambda(\bv):=\bL^{\dagger}\bA\bD^{-1}\bv$. Then
\[
\bD\bx(\bv)-\bv+\bA^\top\blambda(\bv)
=\left[\bv-\bA^\top\bL^{\dagger}\bA\bD^{-1}\bv\right]-\bv+\bA^\top\bL^{\dagger}\bA\bD^{-1}\bv
=\bzero_{n},
\]
which is the stationarity part of the KKT conditions. Thus $(\bx(\bv),\blambda(\bv))$ satisfies the KKT
conditions. As the objective is strictly convex and the constraints are linear equalities,
the KKT conditions are both necessary and sufficient for optimality. It thus follows by uniqueness of the minimizer that $\bx(\bv)=\bx^\star(\bv)$. Since $\bv$ was arbitrary, the equality
\[
\bP\bv=\left[\bD^{-1}-\bD^{-1}\bA^\top\bL^{\dagger}\bA\bD^{-1}\right]\bv
\]
must hold for all $\bv\in\R^n$, and the desired matrix equality follows.
\end{proof}

\begin{proof}[\upshape\textsc{Proof of Corollary \ref{cor:CostFlowDifferentiabilityUnderC2AndCurvature}}]
Let $\cA:=\cS(\obc)$ and $\cI:=\cE\setminus\cA$ be the active and inactive link sets at $\obc$. Since $\obc\in\cC_0$, Theorem~\ref{thm:CostFlowAndValueProperties}\ref{enu:CostSupportConstancy} implies that there exists an open
neighborhood $\cN$ of $\obc$ such that $\cS(\bc)=\cS(\obc)=\cA$ for all $\bc\in\cN$. In particular, $\hx_e(\bc)=0$ for all $e\in\cI$ and all $\bc\in\cN$. Thus the support stability condition \ref{enu:SupportStability} of Theorem~\ref{thm:CostFlowDerivative} holds at $\obc$.

Next, for each $e\in\cA$ we have $\hx_e(\obc)\in\R_{++}$. By assumption, $h_e$ is twice continuously differentiable on $\R_{++}$ and satisfies $\nabla^2 h_e(\xi)>0$ for all $\xi\in\R_{++}$, in particular $\nabla^2 h_e(\hx_e(\obc))>0$. Therefore the local curvature condition \ref{enu:LocalCurvature} of Theorem~\ref{thm:CostFlowDerivative} also holds at $\obc$.

Consequently, Theorem~\ref{thm:CostFlowDerivative} applies at $\obc$ and yields that $\bc\mapsto\hbx(\bc)$ is continuously differentiable (on a neighborhood of $\obc$, hence in particular) at $\obc$, with Jacobian $\nabla\hbx(\obc)$ given by \eqref{eq:FullJacobianBlock} and active-link block $\nabla_{\bc_\cA}\hbx_\cA(\obc)$ given by \eqref{eq:JacobianBlockNullspaceForm}, equivalently
\eqref{eq:JacobianBlockLaplacianForm}.
\end{proof}

\end{document}